%% file: of-main.tex
\documentclass[a4paper,11pt]{article}

\usepackage{verbatim,url,enumerate,color,paralist}
\usepackage{amsmath,amsfonts}
\usepackage{epsfig,amssymb,amstext,amsthm}
\usepackage{psfrag}
\usepackage{algorithm}
\usepackage[noend]{algorithmic}
\usepackage{times,setspace}

\newtheorem{theorem}{Theorem}

\newtheorem{lemma}[theorem]{Lemma}

\newtheorem{corollary}[theorem]{Corollary}
\newtheorem{fact}[theorem]{Fact}

\newcommand{\bfy}{\mathbf{y}}
\newcommand{\bfx}{\mathbf{x}}
\newcommand{\bfw}{\mathbf{w}}
\newcommand{\bfu}{\mathbf{u}}
\newcommand{\bfv}{\mathbf{v}}
\newcommand{\bfp}{\mathbf{P}}

\newcommand{\opt}{\mathcal{OPT}}
\newcommand{\optone}{\mathcal{OPT}_1}
\newcommand{\opttwo}{\mathcal{OPT}_2}

\newcommand{\ty}{\tilde{y}}
\newcommand{\bty}{\tilde{\mathbf{y}}}

\newcommand{\est}{\mathrm{est}}
\newcommand{\rate}{\mathrm{rate}}
\newcommand{\cost}{\mathrm{cost}}

\newcommand{\assign}{\textsc{Assign}($\varGamma$) }

\newcommand{\tv}{\tilde{v}}	
\newcommand{\tw}{\tilde{w}}	
\newcommand{\open}{\mathcal{O}}

\newcommand{\bfpx}{\mathbf{Px}}
\newcommand{\bfcx}{\mathbf{Cx}}

\newcommand{\tA}{\tilde{\mathcal{A}}}

\newcommand{\init}{\mbox{\small{init}}}

\newcommand{\tl}{\tilde{\lambda}}
\newcommand{\tlx}{\tilde{\lambda}(\bfx)}
\newcommand{\tpx}{\tilde{\mathbf{P}}\mathbf{x}}

\newcommand{\tr}{\tilde{r}}

\title{Online Mixed Packing and Covering \thanks{This work was supported in part by NSF grants CCF-0728869 and CCF-1016778.}}
\author{Umang Bhaskar\footnotemark[2] \and Lisa Fleischer\footnotemark[2]}

\footnotetext[2]{Dartmouth College, 6211 Sudikoff Lab, Hanover NH 03755. \{umang, lkf\}@cs.dartmouth.edu.}

\date{\today}
\begin{document}

\maketitle

\begin{abstract}

In many problems, the inputs to the problem arrive over time. As each input is received, it must be dealt with irrevocably. Such problems are \emph{online} problems. An increasingly common method of solving online problems is to solve the corresponding linear program, obtained either directly for the problem or by relaxing the integrality constraints. If required, the fractional solution obtained is then rounded online to obtain an integral solution.

We give algorithms for solving linear programs with mixed packing and covering constraints online. We first consider mixed packing and covering linear programs, where packing constraints $\mathbf{Px} \le \mathbf{p}$ are given offline and covering constraints $\mathbf{Cx} \ge \mathbf{c}$ are received online. The objective is to minimize the maximum multiplicative factor by which any packing constraint is violated, while satisfying the covering constraints. For general mixed packing and covering linear programs, no prior sublinear competitive algorithms are known. We give the first such --- a polylogarithmic-competitive algorithm for solving mixed packing and covering linear programs online. We also show a nearly tight lower bound. 

Our techniques for the upper bound use an exponential penalty function in conjunction with multiplicative updates. While exponential penalty functions are used previously to solve linear programs offline approximately, offline algorithms know the constraints beforehand and can optimize greedily. In contrast, when constraints arrive online, updates need to be more complex. 

We apply our techniques to solve two online fixed-charge problems with congestion. These problems are motivated by applications in machine scheduling and facility location. The linear program for these problems is more complicated than mixed packing and covering, and presents unique challenges. We show that our techniques combined with a randomized rounding procedure can be used to obtain polylogarithmic-competitive integral solutions. These problems generalize online set-cover, for which there is a polylogarithmic lower bound. Hence, our results are close to tight.

\end{abstract}

\thispagestyle{empty}
\setcounter{page}{0}

\newpage

\pagestyle{plain}

\input{of-intro}

\input{of-algo}

\input{of-flcc2}

\paragraph*{Acknowledgements.} We thank Zhenghui Wang for helpful discussions regarding various objectives for the machine scheduling problem, and for the example which appears in Section~\ref{sec:machinebad}.

\bibliographystyle{plain}
\bibliography{online}

\appendix

\input{of-appendix}

\end{document}

%% file: of-intro.tex
\section{Introduction}

In this paper, we give the first online algorithm for general mixed packing and
covering linear programs (LPs) with a sublinear competitive ratio.  The
problem we study is as follows.

\vspace{0.5em}
\noindent \textbf{Online Mixed Packing and Covering (OMPC).}  
Given: a set of \emph{packing
constraints} $\bfpx \le \mathbf{p}$, 
and a set of \emph{covering constraints} 
$\bfcx \ge \mathbf{c}$, with positive coefficients and variables $\bfx$,
such that the packing constraints are known in advance, and the covering
constraints arrive one at a time.  Goal: After the arrival of each
covering constraint, increase $\bfx$ so that the new covering constraint
is satisfied and the amount $\lambda$ by which we must multiply $\mathbf{p}$
to make the packing constraints feasible is as small as possible.
\vspace{0.5em}

Mixed packing and covering problems model a wide range of problems
in combinatorial optimization and operations research.  These problems
include facility location, machine scheduling, and circuit routing. In these problems, requests for resources such as bandwidth or processing time are received over time, or \emph{online}, whereas the set of resources is known offline. As each request arrives, we must allocate resources to satisfy the request. These allocations are often impossible or costly to revoke. The resources correspond to packing constraints in our setting and are known offline. Requests correspond to covering constraints, and arrive online. 
 The performance of an online algorithm is 
measured by the \emph{competitive ratio}, defined as the worst 
case ratio of the value of the solution obtained by the online algorithm to 
the value obtained by the optimal offline algorithm which has as its input 
the entire sequence of requests. The worst case ratio is over all possible 
sequences of inputs. 

Many techniques to solve integer problems online first obtain
a fractional solution, and then round this to an integer solution
~\cite{AlonAABN06, AlonAABN09,  BansalBN07, BansalBN08}. 
The first step involves solving a linear program relaxation of the original problem online.
In fact, this can be a bottleneck step in obtaining a good competitive ratio.
Thus, our algorithm for online mixed packing and covering 
can provide an important first step in obtaining
good online solutions to several combinatorial problems.  We demonstrate
the power of our ideas by extending them to give the first online
algorithms with sublinear competitive ratios for a number of fixed-charge problems with capacity constraints.  For these problems, we first solve the linear program relaxation online, 
and then use known randomized rounding techniques to obtain an integer solution online.

\vspace{0.5em}
\noindent \textbf{Applications.} We use our techniques to study two problems with fixed-charge and congestion costs. Both fixed-charge problems and congestion problems are widely studied offline and online; we discuss specific applications and references below. In general, fixed charges are used to model one-time costs such as resource purchases or installation costs, while congestion captures the load on any resource. In machine scheduling, for example, the makespan can be modelled as the maximum congestion by setting each resource to be a machine and setting unit capacity for each machine. 

\vspace{0.5em}
\noindent \textbf{Application 1: Unrelated Machine Scheduling with Start-up Costs (UMSC).}  Given offline: a set of machines $\{1, \ldots, m\}$
with start-up cost $c_i$ for machine $i$. Jobs arrive online, and job $j$ requires $p_{ij}$ time to be processed
on machine $i$. Goal: when a job $j$ arrives, determine
whether to ``open'' new machines by paying their start-up cost, 
and then assign the job to one of the open machines, so that the
sum of the \emph{makespan} of the schedule --- the maximimum over
machines of the processing times of the jobs assigned to it ---
and the sum of start-up costs is minimized.
\vspace{0.5em}

The problem of scheduling jobs to minimize the makespan and the fixed charges is studied both offline~\cite{Fleischer10,KhullerLS10} and online~\cite{DosaT10, ImrehN99,Imreh09}. The problem is motivated by reducing energy consumption in large data centers, such as those used by Google and Amazon~\cite{BirmanCR09,KhullerLS10}. The energy consumption of a large data center is estimated to exceed that of thousands of homes, and the energy costs of these centers is in the tens of millions of dollars~\cite{QureshiWBGM09}, hence algorithms that focus on reducing energy consumption are of practical importance. The inclusion of a fixed charge models the cost of starting up a machine. Thus machines do not need to stay on, and can be started when the load increases. Bicriteria results for the offline problem are given in~\cite{KhullerLS10} and~\cite{Fleischer10} using different techniques. We show strong lower bounds for bicriteria results in the online setting, and therefore focus on algorithms for the sum objective. For the online problem with identical machines,~\cite{DosaT10, ImrehN99} give constant-competitive algorithms for the sum objective. These are extended to the case where machines have speed either 1 or $s$, with more general costs for the machines in~\cite{Imreh09}. 

\vspace{0.5em}
\noindent \textbf{Application 2: Capacity Constrained Facility Location (CCFL).} Given offline: a set of facilities $\mathcal{F}$ with fixed-charge $c_i$ and capacity $u_i$ for each facility $i$ in $\mathcal{F}$. Clients arrive online, and each client $j$ has an assignment cost $a_{ij}$ and a demand $p_{ij}$ on being assigned to facility $i$. Goal: when client $j$ arrives, determine whether to open new facilities by paying their fixed charge, and then assign the client to an open facility, so that the sum of the maximum congestion of any facility, total assignment costs, and the total fixed charges paid for opened facilities is minimized. The congestion of a facility is the ratio of the sum of the loads of clients assigned to the facility to the capacity of the facility.
\vspace{0.5em}

For online facility location without capacity constraints, a $\Theta \left(\frac{\log n}{\log \log n}\right)$-competitive ratio is possible when the assignment costs form a metric~\cite{Fotakis08}, and a $O(\log m \log n)$-competitive ratio is possible when assignment costs are non-metric~\cite{AlonAABN06}, with $n$ clients and $m$ facilities. Capacitated facility location is a natural extension to the problem. In the offline setting, constant-factor approximation algorithms are known for both facility location with soft capacities --- when multiple facilities can be opened at a location --- and hard capacities --- when either a single facility or no facility is opened at each location~\cite{MahdianYZ06,ZhangCY05}. Our problem is a variant of non-metric soft-capacitated facility location where instead of minimizing the cost of installing multiple facilities at a location, we minimize the load on the single facility at each location, in addition to fixed-charge and assignment costs.

\vspace{0.5em}
\noindent \textbf{Our Results.}
We give polylogarithmic competitive ratios for the problems discussed.  Our results are the first sublinear guarantees for these problems.

\begin{itemize}
 \item For OMPC:
 \begin{itemize}
  \item A deterministic $O(\ln m \ln(d\rho \kappa))$-competitive algorithm, 
where $m$ is the number of packing constraints, $d$ is the maximum number of variables in any constraint, $\rho$ is the ratio of the maximum to the minimum non-zero packing coefficient and $\kappa$ is the ratio of the maximum to the minimum non-zero covering coefficient (Section~\ref{sec:algo}). If all coefficients are either 0 or 1, this gives a $O(\ln m \ln d)$-competitive algorithm. 
  \item A lower bound of $O(\ln m \ln (d/\ln m))$ for any deterministic algorithm for OMPC. Our algorithm for OMPC is thus nearly tight (Section~\ref{sec:lower}).
 \end{itemize}
 \item For CCFL and UMSC:
 \begin{itemize}
  \item A randomized $O(\ln (m n \rho) \ln^2 (mn))$-competitive algorithm for CCFL, where $m$ and $n$ are the number of facilities and clients respectively, and $\rho$ is is the ratio of the maximum to the minimum total cost of assigning a single client (Section~\ref{sec:ccfl}). We obtain the same competitive ratio for UMSC, where $m$ and $n$ are the number of machines and jobs respectively, and $\rho$ is the ratio of the maximum to the minimum total cost of assigning a single job.
  \item A lower bound for bicriteria results for CCFL: even if the maximum congestion $T$ is given offline, no deterministic online algorithm can obtain a fractional solution with maximum congestion $o(m)T$ and fixed-charge within a polylogarithmic factor of the optimal (Section~\ref{sec:machinebad}). This lower bound also holds for UMSC , where $T$ is the makespan.
 \end{itemize}
\end{itemize}

Since each of our applications include fixed-charges as part of the objective, they generalize online set cover. In UMSC, for example, set cover is obtained by setting the processing times to be either zero or infinity. Since the makespan in any bounded solution to the problem is now zero, this reduces the problem to covering jobs with machines to minimize the sum of machine startup costs. Online set cover has a lower bound of $\Omega(\log m \log n)$ on the competitive ratio assuming BPP $\neq$ NP~\cite{AlonAG09}.  Thus, our results for UMSC and CCFL are tight modulo a logarithmic factor.

\vspace{0.5em}
\noindent \textbf{Our Techniques.}
Our techniques for online mixed packing and covering are based on a novel extension of multiplicative weight updates. We replace the packing constraints in our problem with an exponential penalty function that gives an upper bound on the violation of any constraint. When a covering constraint arrives, the increment to any variable is  inversely proportional to the rate of change of this penalty function with respect to the variable. We use a primal-dual analysis to show that this technique, combined with a doubling approach used in previous online algorithms, gives the required competitive ratio. While exponential potential functions are widely used for offline algorithms and machine learning, e.g.~\cite{AroraHK05}, our work is the first to use an exponential potential function to drive multiplicative updates that yield provably good competitive ratios for \emph{online} algorithms.

Our work is closely related to work on solving pure packing packing and pure covering linear programs online, and Lagrangean-relaxation techniques for solving linear programs approximately offline. Multiplicative weight updates are used in~\cite{BuchbinderN05} to obtain $O(\log n)$-competitive fractional solutions for covering linear programs when the constraints arrive online. In~\cite{BuchbinderN05}, the cost is a simple linear function of the variables. The update to each variable is inversely proportional to the sensitivity of the cost function relative to the variable, given by the variable's coefficient in the cost function. In our problem, however, the cost is the maximum violation of any packing constraint. The cost function is thus nonlinear, and since its sensitivity relative to a variable changes, it is not apparent how to extend the techniques from~\cite{BuchbinderN05}. We use an exponential potential function to obtain a differentiable approximation to this nonlinear cost. For each variable, our updates depend on the sensitivity of this potential function relative to the variable. In addition to the primal-dual techniques in~\cite{BuchbinderN05}, a key step in our analysis is to obtain bounds on the rate of change of this potential function.

A large body of work uses Lagrangean-relaxation techniques to obtain approximate algorithms for solving LPs offline, e.g.,~\cite{PlotkinST91,Young01}. In these papers, the constraints in the LP are replaced by an exponential penalty function. In each update, the update vector for the variables minimizes the change in the penalty function. In this sense, the updates in these offline algorithms are greedy. Since the constraints are available offline, this gives $\epsilon$-approximate solutions. In our case, since covering constraints arrive online, greedy algorithms perform very poorly, and we must use different techniques. We use an exponential penalty function similar to offline algorithms. However, our updates are very different. Instead of a greedy strategy as used in~\cite{PlotkinST91,Young01}, we hedge our bets and increment all variables that appear in the covering constraint. The increment to each variable is inversely proportional to its contribution to the penalty function.

For fixed-charge problems with capacity constraints, we solve the corresponding linear programs online and, for our applications, round the fractional solutions to obtain integral solutions online. The linear programs for these problems are significantly more complicated than mixed packing and covering. We combine our techniques for mixed packing and covering with a more complex doubling approach to obtain fractional solutions, and adapt randomized rounding procedures used previously offline for machine scheduling~\cite{KhullerLS10} and online for set cover~\cite{BuchbinderN09_2} to obtain integral solutions.

\vspace{1em}
\noindent \textbf{Other Related Work.}
Multiplicative updates are used in a wide variety of contexts.  They
are used in both offline approximation algorithms for packing and
covering problems~\cite{BienstockI06,Fleischer00,GargK07,GrigoriadisK94,GrigoriadisK95,GrigoriadisK96,KleinPST94,KoufogiannakisY07,LeightonMPSST95,PlotkinST91,ShahrokhiM90,Young01}, and online approximations for pure packing or
pure covering problems based on linear programs such as
set cover~\cite{BuchbinderN09}, caching~\cite{BansalBN08}, paging~\cite{BansalBN07}, and ad allocations~\cite{BuchbinderJN07}. Both offline and online, these algorithms are
analyzed using a primal-dual framework.  
Multiplicative updates are used earlier~\cite{AlonAABN06,AlonAABN09} to implicitly solve a linear 
program online for various network optimization problems. The fractional solution obtained was rounded online to obtain an integral solution. Multiplicative weight updates also have a long history in learning theory; these results are surveyed in~\cite{AroraHK05}.

Our work studies the worst-case behaviour of our algorithms assuming 
adversarial inputs. A large body of work studies algorithms for online 
problems when the inputs are received as the result of a stochastic 
process. Two common models studied in the literature are (1) when the 
inputs are picked from a distribution (either known or unknown), and 
(2) when an adversary picks the inputs, but the inputs are presented to 
the algorithm in random order. The adwords and display ads problems can 
be modeled as packing linear programs with variables arriving online. 
A number of papers give algorithms for these problems assuming stochastic 
inputs; some of these results are presented in~\cite{Devanur11, FeldmanHKMS10}.

%% file: of-algo.tex
\section{Online Mixed Packing and Covering}
\label{sec:algo}

In this section, we consider mixed packing and covering linear programs. A mixed packing and covering linear program has two types of constraints: covering constraints of the form $\bfcx \ge \mathbf{c}$, and packing constraints of the form $\bfpx \le \mathbf{p}$. We normalize the constraints so that the right side of each constraint is 1. Our objective is to obtain a solution $\bfx$ that minimizes the maximum amount by which any packing constraint is violated. Thus, our problem is to obtain a solution to the following linear program:

\begin{equation}
\min \lambda ~ \mbox{s.t. } \bf{Cx} \ge \mathbf{1}, ~ \bf{Px} \le \lambda, ~ \bf{x}, \lambda \ge 0 \, . \label{lp:mpcprimal}
\end{equation}

The packing constraints are given to us initially, and the covering constraints are revealed one at a time. Our online algorithm assigns fractional values to the variables. As covering constraints arrive, the variable values can be increased, but cannot be decreased. 

For a vector $\bf{v}$, we use both $v_i$ and $(v)_i$ to denote its $i$th component. We use $[n]$ to denote the set $\{1$, $2$, $\dots$, $n\}$. The vector of all ones and all zeros is denoted by $\mathbf{1}$ and $\mathbf{0}$, respectively.

The number of variables, number of packing constraints, and number of covering constraints in the linear program are denoted by $n$, $m$, and $m_c$ respectively. We use $d$ to denote the maximum number of variables in any constraint. We define $\rho := \max_{k,j} p_{kj} / \min_{k,j:p_{kj} > 0} p_{kj}$ and similarly $\kappa = \max_{i,j} c_{ij}/ \min_{i,j:c_{ij} > 0} c_{ij}$. The value of $\kappa$ is used only in the analysis of the algorithm; we do not need to know its value during execution. Define $\kappa_1 := \max_j c_{1j}$, i.e., $\kappa_1$ is the maximum coefficient in the first covering constraint to arrive. Define $d_1$ as the maximum number of variables in any packing constraint, and the first covering constraint. Define $\mu := 1 + \frac{1}{3 \ln (e m)}$, and $\sigma := e^2 \ln (\mu d^2 \rho \kappa)$. Here, $e$ is the base of the natural logarithm.

We use $OPT$ to denote the optimal value of $\lambda$ given $\mathbf{P}$ and $\mathbf{C}$, hence $OPT$ is the value returned by the optimal offline algorithm. 

In order to analyze our algorithm, we consider the dual of~(\ref{lp:mpcprimal}) as well:

\begin{equation}
\max \sum_i y_i  ~ \mbox{ s.t. } {\bf C^T y}  \le {\bf P^T z}, ~~ \sum_{k=1}^{m} z_k  \le 1, ~~ \bf{y}, \, \bf{z} \ge 0 \label{lp:mpcdual}
\end{equation}

\subsection{An Algorithm for Mixed Packing and Covering Online}
\label{sec:algoalgo}

We now give an algorithm for solving OMPC and show that it is $O(\log (d\rho \kappa) \log m)$-competitive. We assume in the following discussion that we are given a scaling parameter $\varGamma \geq \max_{k,j} p_{kj}/(d_1 \rho \kappa_1)$, which is used to scale the matrix of packing coefficients $\bfp$. In Theorem~\ref{thm:main}, we show that if $2 OPT$ $\ge \frac{\varGamma}{4\sigma}$ $\ge OPT$ then our algorithm yields the stated competitive ratio. Without this estimate $\varGamma$, we can use a ``doubling procedure'' commonly used in online algorithms, which increases the competitive ratio obtained by a factor of 4 (Section~\ref{sec:mpcdoubling}).

Given a vector $\bfx$, let $\lambda(\bfx) := \max_{k \in [m]} (\bfpx)_k$. For a given scaling parameter $\varGamma$, let $\tilde{\bfp} := \bfp/\varGamma$, $\tilde{p}_{kj} := p_{kj}/\varGamma$ and $\tlx := \max_{k \in [m]} (\tpx)_k$. Let $\est(\bfx) := \ln \left( \sum_{k \in [m]} \exp(\tpx)_k \right)$ be an estimate of $\tlx$, and note that $\max_k (\tpx)_k \le \est(\bfx) \le \max_k (\tpx)_k + \ln m  $. For each variable $x_j$, define 

\begin{equation}
  \rate_j({\bf x}) ~ := ~ \displaystyle \frac{\partial \est({\bf x})}{\partial x_j} ~ = ~ \displaystyle \frac{\sum_{k \in [m]} \tilde{p}_{kj} \exp(\tpx)_k}{\sum_{k \in [m]} \exp(\tpx)_k} \, . \label{eqn:rj}
\end{equation}

 Our algorithm is given as Algorithm~\ref{algo:main}. Upon receiving the first constraint, we initialize $x_j \leftarrow 1/(d_1^2 \rho \kappa_1)$ for all $j \in [n]$. We also initialize a counter variable $l \leftarrow 0$.

 When a covering constraint $({\bf Cx})_i \ge 1$ arrives it gets assigned a new dual variable $y_i$, and the variables are incremented as described. The dual variables $\mathbf{y}$ are used only in the analysis.
 
 For covering constraint $i$, define 
 
 \begin{equation}
  \epsilon_i(\bfx) := (\mu - 1) \min_{j:c_{ij} > 0} \rate_j(\bfx)/ c_{ij} \, , \label{eqn:epsilon}
  \end{equation}
  
  \noindent so that for all $j \in [n]$, $\epsilon_i(\bfx) c_{ij}/\rate_j(\bfx) \le \mu - 1$. In line~\ref{algo:line:xj}, each variable $x_j$ gets increased by at most a factor of $\mu$, and at least one variable gets incremented by a factor of $\mu$.

\begin{algorithm}[h]
\caption{\textsc{MPC-Approx}: Upon arrival of $i$th covering constraint:}
\label{algo:main}
\begin{algorithmic}[1]
  \STATE When first constraint arrives, initialize $x_j \leftarrow 1/(d_1^2 \rho \kappa_1)$ for all $j \in [n]$, and $l \leftarrow 0$.
  \STATE Upon arrival of $i$th covering constraint:
  \WHILE{$({\bf Cx})_i < 1$}
    \STATE $l \leftarrow l + 1$, $\bfx^l \leftarrow x$
    \STATE $\forall j$, $\rate_j \leftarrow \rate_j({\bf x^l})$ \hspace{0.5in} /* defined in~(\ref{eqn:rj}) */
    \STATE $\displaystyle \epsilon_i \leftarrow \epsilon_i(\bfx^l)$ \hspace{1.15in} /* defined in~(\ref{eqn:epsilon}) */
    \FOR{$j \in [n]$}
      \STATE $\displaystyle x_j \leftarrow x_j \left( 1 + \epsilon_i \frac{c_{ij}}{\rate_j} \right)$ \label{algo:line:xj}
    \ENDFOR
  \STATE $y_i \leftarrow y_i + e \epsilon_i$ \hspace{1in} /* for analysis */
  \STATE \textbf{if } $\tl(\bfx) \ge 3 \ln (em)$ \textbf{then return } FAIL \label{line:mpcfail}
\ENDWHILE
\end{algorithmic}
\end{algorithm}

A single iteration of the while loop is a \emph{phase}, indexed by $l$, and the first phase is phase 0. The value of the variables before they are incremented in phase $l$ is $\bfx^l$. $\bfx^0$ denotes the values after initialization. For covering constraint $i$, $L_i$ is the indices of the phases executed from its arrival until $(\bfcx)_i \ge 1$, and $L := \cup_i L_i$. 

We first show upper bounds on values attained by the variables, and on the running time.

\begin{lemma}
 During the execution of the algorithm, for any $j \in [n]$, $x_j \le \mu/\min_{i:c_{ij} > 0} c_{ij}$.
\label{lem:xjbound}
\end{lemma}

\begin{proof}
For any $x_j$, if $\min_{i:c_{ij} > 0} c_{ij} x_j \ge 1$, then $x_j$ will not be incremented further in any phase since any covering constraint $i$ with $c_{ij} > 0$ must already be satisfied. Thus, since the value of any variable increases by at most a factor of $\mu$ in a phase, $x_j \le \mu/\min_{i:c_{ij} > 0} c_{ij}$.
\end{proof}

\begin{lemma}
 \textsc{MPC-Approx} executes $O(n \ln (\mu d^2 \rho \kappa) \ln m)$ phases, and each phase takes time $O(m n)$.
 \label{lem:runtime}
\end{lemma}

\begin{proof}
In each phase, the value of at least one variable gets incremented by a factor of $\mu$. Each variable has an initial value of $1/(d_1^2 \rho \kappa_1)$.  Let $n_j$ be the number of phases in which $x_j$ gets increased by a factor of $\mu$. Then $x_j \ge \mu^{n_j}/(d_1^2 \rho \kappa_1)$. Since by Lemma~\ref{lem:xjbound}, $x_j \le \mu/\min_{i:c_{ij} > 0} c_{ij}$, $n_j \le \log_\mu (\mu d_1^2\rho\kappa_1/\min_{i:c_{ij} > 0} c_{ij})$. Observing that for all $j$, $\kappa_1/\min_{i:c_{ij} > 0} c_{ij} \le \kappa$ and $d_1 \le d$, it follows that each variable can be increased by $\mu$ in at most $\log_\mu (\mu d^2 \rho \kappa)$ phases. Since in each phase at least one variable increases by a factor of $\mu$, the number of phases is at most $n\log_\mu (\mu d^2 \rho \kappa)$ $= n \ln (\mu d^2 \rho \kappa)/\ln \mu$. Since $\ln (1+x) \ge x/e$ for $0 \le x \le 1$, $\ln \mu \ge 1/(3e \ln (e m))$. Thus the number of phases is at most $O(n \log (\mu d^2 \rho \kappa) \log m)$. In each phase, $\tpx$ can be computed in $O(m n)$ time; then $\est(\bfx)$ and each $\rate_j(\bfx)$ can be computed in time $O(m)$. Thus each phase takes time $O(m n)$.
\end{proof}

Our proof of the competitive ratio follows from a primal-dual analysis. We show in Corollary~\ref{cor:primaldual2} that $\ln m + OPT/\varGamma$ plus the value of the dual objective maintained by the algorithm is an upper bound on the primal objective maintained by the algorithm. Lemmas~\ref{lem:dualsumz} and~\ref{lem:dualyz} show how the dual variables maintained by the algorithm can be scaled down to obtain feasible dual values. We show in Theorem~\ref{thm:main} that together these prove the bound on the competitive ratio.

We first show that the initialization of the variables ensures that $\est(\bfx^0)$ does not exceed $ \frac{OPT}{\varGamma} + \ln m$.

\begin{lemma}
 For the variables as initialized, $\tl(\bfx^0) \le OPT/\varGamma$, and hence $\est(\bfx^0) \le \frac{OPT}{\varGamma} + \ln m$.
\label{lem:init}
\end{lemma}

\begin{proof}
Let $x_j^*$ be the values for the variables in an optimal solution. After the first covering constraint is received, $1 \le \sum_j c_{1j} x_j^* \le \max_r c_{1r} \sum_j x_j^*$. Since the first covering constraint has at most $d_1$ variables, there exists variable $x_b^* \ge 1/ (d_1 \max_r c_{1r})$, and hence 

\begin{equation*} 
 OPT ~ = ~ \max_{k \in [m]} (\bfpx^*)_k ~ \ge ~ \min_{k,j:p_{kj} > 0} p_{kj} x_b^* ~ \ge ~ \min_{k,j:p_{kj} > 0} p_{kj} / (d_1 \max_r c_{1r}) ~ = ~ \min_{k,j:p_{kj} > 0} p_{kj} / (d_1 \kappa_1) \, .
\end{equation*}

\noindent Using $\rho = \max_{k,j} p_{kj} / \min_{k,j:p_{kj} > 0} p_{kj}$,

\begin{equation}
 OPT ~ \ge ~ \max_{k,j:p_{kj} > 0} p_{kj} / (d_1 \rho \kappa_1) ~ = ~ \varGamma \max_{k,j:p_{kj} > 0} \tilde{p}_{kj} / (d_1 \rho \kappa_1)\, .
\label{eqn:initopt}
\end{equation}

\noindent Our algorithm initializes $x_j^0 = 1/(d_1^2 \rho \kappa_1)$, and hence 

\begin{equation}
 \tl(\bfx^0) ~ = ~ \max_{k \in [m]} (\tpx^0)_k ~ \le ~ d_1 \max_{k,j}\tilde{p}_{kj} / (\varGamma d_1^2 \rho \kappa_1) ~ \le ~ \max_{k,j} \tilde{p}_{kj} / (\varGamma d_1 \rho \kappa_1) \stackrel{(\ref{eqn:initopt})}{\le} OPT/\varGamma \, ,
\label{eqn:initopt2}
\end{equation}

\noindent where the first inequality is because any packing constraint has at most $d_1$ variables. Thus, $\est(\bfx^0)$ $\le \tl(\bfx^0) + \ln m $ $\le (OPT/\varGamma) + \ln m $, proving the lemma.
\end{proof}

\begin{corollary}
 If $\varGamma \geq \max_{k,j} p_{kj}/(d_1 \rho \kappa_1)$, then $\tl(\bfx^l) \le 3 \ln (e m)$ at the beginning of any phase $l$.
 \label{cor:init2}
\end{corollary}

\begin{proof}
Since $\varGamma \geq \max_{k,j} p_{kj}/(d_1 \rho \kappa_1)$, by~(\ref{eqn:initopt2}), $\tl(\bfx^0) \le 1$. Thus the lemma is satisfied for the first phase. For any phase $l > 0$, the algorithm would have failed at the end of phase $l-1$ if $\tl(\bfx) \ge 3 \ln (em)$. Since the algorithm did not fail in phase $l-1$, in any phase $l$, $\tl(\bfx^l) \le 3 \ln (em)$. 
\end{proof}

\begin{lemma}
 If $\varGamma \geq \max_{k,j} p_{kj}/(d_1 \rho \kappa_1)$, the increase in the dual objective $\sum_i y_i$ is an upper bound on the increase in $\est(\bfx)$ in every phase. \label{lem:primaldual}
\end{lemma}

\begin{proof}
Let $\est^l$ and $\est^{l+1}$ denote the values of $\est({\bf x})$ before and after the variables are incremented in phase $l$, respectively. We will show that $\est^{l+1} - \est^l \le e \epsilon_i$, which is the increase in $\sum_i y_i$ in phase $l$.
 
 Let $\bfx^l$ and $\bfx^{l+1}$ be the values of $\bfx$ before and after being incremented in phase $l$. For each $x_j$, let $g_j(t) := x_j^l + (x_j^{l+1} - x_j^l)t$ for $0 \le t \le 1$. Note that $g_j(0) = x_j^l$ and $g_j(1) = x_j^{l+1}$. Define $\mathbf{g}(t) = (g_1(t), g_2(t), \dots, g_m(t))$. With some abuse of notation, any function of $\bfx$, say $h(\bfx)$, can be viewed as a function of $t$, with $h(t) := h(\mathbf{g}(t))$. Thus, the functions $\est(\bfx)$ and $\rate_j(\bfx)$ can be written as functions of $t$: $\est(t) = \ln \sum_{k \in [m]} \exp (\mathbf{\tilde{P}g}(t))_k$, and

\begin{equation}
 \rate_j(t) ~ = ~ \rate_j(\mathbf{g}(t)) ~ = ~ \frac{\partial \est(t)}{\partial g_j(t)} ~ = ~ \frac{\sum_{k\in [m]} \tilde{p}_{kj} \exp({\bf \tilde{P}g}(t))_k}{\exp(\est(t))}\, .
\label{eqn:rjt}
\end{equation}

\noindent We use these alternate expressions in the remainder of the proof. By the chain rule,

 \[
  \displaystyle \frac{d \est(t)}{dt} ~ = ~ \sum_{j=1}^n \frac{\partial \est(t)}{\partial g_j(t)} \frac{dg_j(t)}{dt} \, ,
 \]

\noindent and hence,

\begin{equation}
  \est^{l+1} - \est^l ~ = ~ \displaystyle \int_{t=0}^1 \frac{d \est(t)}{dt} dt 
    ~ =~ \displaystyle  \sum_{j=1}^n \int_{t=0}^1 \rate_j(t) \frac{dg_j(t)}{dt} dt\, . 
 \label{eqn:deltaq}
\end{equation}

\noindent In a phase, each variable is incremented by at most a factor of $\mu$. Therefore $\bfx^{l+1} \le \mu \bfx^l$. Then since $\varGamma \geq \max_{k,j} p_{kj}/(d_1 \rho \kappa_1)$, by Corollary~\ref{cor:init2}, $\tl(\bfx^l) \le 3 \ln (e m)$ in any phase $l$. Thus $\rate_j(t) \le e \rate_j(0)$ for $0 \le t \le 1$ by Lemma~\ref{lem:rjbound} (in Appendix). Hence

\[
  \est^{l+1} - \est^l ~ \le ~ e \sum_{j=1}^n \rate_j(\bfx^l) \int_{t=0}^1 \frac{dg_j(t)}{dt} dt ~ = ~ e \displaystyle \sum_{j=1}^n \rate_j(\bfx^l)  (x_j^{l+1} - x_j^l) \, .
\]

\noindent Since in phase $l$ each variable $x_j$  gets multiplied by $1 + \epsilon_i(\bfx^l) \frac{c_{ij}}{\rate_j(\bfx^l)}$, 

\[
  \est^{l+1} - \est^l ~ \le ~ e  \epsilon_i(\bfx^l) \displaystyle \sum_{j=1}^n \rate_j(\bfx^l) \frac{c_{ij} x_j^l}{\rate_j(\bfx^l)} ~ = ~ e \epsilon_i(\bfx^l) \displaystyle \sum_{j=1}^n c_{ij} x_j^l ~ \le ~ e \epsilon_i(\bfx^l) \\
\]

\noindent where the last inequality follows since, on entering the for loop, $(\bfcx)_i < 1$. Since $ e \epsilon_i(\bfx^l)$ is the increase in the dual objective, this proves the lemma.
\end{proof}

\begin{corollary}
  If $\varGamma \geq \max_{k,j} p_{kj}/(d_1 \rho \kappa_1)$, then $\sum_i y_i \ge \est(\bfx) - \frac{OPT}{\varGamma} - \ln m $.
\label{cor:primaldual2}
\end{corollary}

\begin{proof}
By Lemma~\ref{lem:primaldual}, the increase in $\sum_i y_i$ is an upper bound on the increase in $\est(\bfx)$, thus $\sum_i y_i \ge \est(\bfx) - \est(\bfx^0)$. By Lemma~\ref{lem:init}, $\est(\bfx^0) \le \ln m + \frac{OPT}{\varGamma}$, and hence $\sum_i y_i \ge \est(\bfx) - \ln m - \frac{OPT}{\varGamma}$.
\end{proof}

We now show that the dual variables do not violate the dual constraints by much. We choose the dual variable $z_k$ corresponding to each packing constraint $k \in [m]$ as

\begin{equation}
 \displaystyle z_k ~ := ~ \max_{l \in L} \frac{\exp((\tpx^l)_k)}{\exp(\est(\bfx^l))} \label{eqn:zdef}
\end{equation}

\begin{lemma}
 For $\mathbf{z}$ as defined in~(\ref{eqn:zdef}), $ \displaystyle \sum_{k \in [m]} z_k ~ \le ~ \ln (e m) + \max_{l \in L} \tl(\bfx^l)$.
 \label{lem:dualsumz}
\end{lemma}

\begin{proof}
 For each packing constraint $k$, let $\phi(k) := \arg \max_l \frac{\exp((\tpx^l)_k)}{\exp(\est(\bfx^l))}$. Thus $z_k$ attains its value in phase $\phi(k)$. We index the packing constraints so that $\phi(1) \le \phi(2) \le \dots \le \phi(m)$. Then for any $r, k \in [m]$ with $k \ge r$ so that $\phi(k) \ge \phi(r)$, we have $(\tpx^{\phi(k)})_r \ge (\tpx^{\phi(r)})_r$ since the variables $x$ are increasing. Thus,
 
\begin{eqnarray}
  \exp(\est(\bfx^{\phi(k)})) ~ = ~ \displaystyle \sum_{r \in [m]} \exp((\tpx^{\phi(k)})_r )
  ~ \ge ~ \displaystyle \sum_{r \le k} \exp((\tpx^{\phi(k)})_r) 
  ~ \ge ~ \displaystyle \sum_{r \le k} \exp((\tpx^{\phi(r)})_r) \label{eqn:sumz1} \, .
\end{eqnarray}

\noindent Substituting~(\ref{eqn:sumz1}) into~(\ref{eqn:zdef}) yields

\begin{eqnarray*}
 z_k ~ = ~ \displaystyle \frac{\exp((\tpx^{\phi(k)})_k)}{\exp(\est(\bfx^{\phi(k)}))}
 ~ \le ~ \displaystyle \frac{\exp((\tpx^{\phi(k)})_k)}{\sum_{r \le k} \exp(\tpx^{\phi(r)})_r} \, . \end{eqnarray*}

\noindent Then by Lemma~\ref{lem:genineq2} in the appendix, with $a_k = \exp((\tpx^{\phi(k)})_k)$,

\begin{eqnarray}
 \displaystyle \sum_{k \in [m]} z_k ~ \le ~  1 + \ln \left(\frac{\sum_{k \in [m]} \exp((\tpx^{\phi(k)})_k)}{\exp(\tpx^{\phi(1)})_1}\right) ~ \le ~ 1 + \ln m + \max_l \tl(\bfx^l) \label{eqn:sumz2} \, ,
\end{eqnarray}

\noindent where the last inequality follows since $\exp((\tpx^{\phi(1)})_1) \ge 1$ and $(\tpx^l)_k \le \tl(\bfx^l)$ by definition of $\tl$. 
\end{proof}

The next lemma tells us how much we must scale the dual solution obtained by the algorithm to obtain a dual feasible solution. 

\begin{lemma}
 For any $j \in [n]$, $(\mathbf{C^Ty})_j ~ \le ~ \displaystyle (\mathbf{P^Tz})_j \, \frac{\sigma}{\varGamma}$ .
 \label{lem:dualyz}
\end{lemma}

\begin{proof}
 Consider a phase $l$ executed upon arrival of a covering constraint $i$. In this phase, $y_i$ gets incremented by $e \epsilon_i(\bfx^l)$. This increment occurs in every phase in $L_i$. Hence

\begin{eqnarray}
 \displaystyle (\mathbf{C^Ty})_j ~ = ~ \sum_{i \in [m_c]} c_{ij} y_i ~ = ~ e \sum_{i \in [m_c]} c_{ij} \sum_{l \in L_i} \epsilon_i(\bfx^l) \, . \label{eqn:yzbegin}
\end{eqnarray}

\noindent By Lemma~\ref{lem:xjbound}, $x_j \le \mu/\min_{i:c_{ij}>0} c_{ij}$. Further, since the initial value of  $x_j$ is $1/(d_1^2 \rho\kappa_1)$ and is multiplied by $\left(1+\epsilon_i \frac{c_{ij}}{\rate_j}\right)$ in every phase,  for all $j \in [n]$,

\begin{eqnarray*}
 \frac{\mu}{\min_{i:c_{ij}>0} c_{ij}} ~ \ge ~ x_j ~ = ~ \displaystyle \frac{1}{d_1^2 \rho \kappa_1} \prod_{i \in [m_c]} \prod_{l \in L_i} \left( 1 + \epsilon_i(\bfx^l) \frac{c_{ij}}{\rate_j(\bfx^l)} \right) ~ \ge ~ \displaystyle \frac{1}{d_1^2 \rho \kappa_1} \prod_{i \in [m_c]} \prod_{l \in L_i} \exp \left( \epsilon_i(\bfx^l) \frac{c_{ij}}{e \, \rate_j(\bfx^l)} \right) \, .
\end{eqnarray*}

\noindent where the last inequality is since $\epsilon_i c_{ij}/\rate_j \le 1/(3 \ln (e m)) \le 1$ and for $0 \le a \le 1$, $e^{a/e} \le 1 + a $. Multiplying on both sides by $d_1^2 \rho \kappa_1$, taking the natural log, and reversing the inequality,

\begin{eqnarray*}
 \displaystyle \sum_{i \in [m_c]} \sum_{l \in L_i} \epsilon_i(\bfx^l) \frac{c_{ij}}{e \rate_j(\bfx^l)} & \le & \ln \left( \frac{\mu d_1^2 \rho \kappa_1}{\min_{i:c_{ij}>0} c_{ij}} \right) ~ \le ~ \ln \left( \mu d_1^2 \rho \kappa \right)
\end{eqnarray*}

\noindent and multiplying both sides by $e \cdot \max_{l \in L} \rate_j(\bfx^l)$,

\begin{eqnarray}
 \displaystyle \sum_{i \in [m_c]} \sum_{l \in L_i} \epsilon_i(\bfx^l) c_{ij} & \le & e \max_{l \in L} \rate_j(\bfx^l)  \ln \left( \mu d_1^2 \rho \kappa \right) \, .\label{eqn:yz1}
\end{eqnarray}

Thus from~(\ref{eqn:yzbegin}) and~(\ref{eqn:yz1}), $(\mathbf{C^Ty})_j \le e^2 \max_{l \in L} \rate_j(\bfx^l)  \ln \left( \mu d_1^2 \rho \kappa \right)$ $\le \sigma \max_{l \in L} \rate_j(\bfx^l)$. We will now show that $\max_{l \in L} \rate_j(\bfx^l) \le (\mathbf{\tilde{P}^Tz})_j$, completing the proof. This follows since

\[
 \displaystyle \max_{l \in L} \rate_j(\bfx^l) 
 ~ = ~ \max_{l \in L} \frac{\sum_{k \in [m]} \tilde{p}_{kj}\exp((\tpx^l)_k)}{\exp(\est(\bfx^l))}
 ~ \le ~ \sum_{k \in [m]} \tilde{p}_{kj} \max_{l \in L} \frac{\exp((\tpx^l)_k)}{\exp(\est(\bfx^l))}
 ~ = ~ \sum_{k \in [m]} \tilde{p}_{kj} z_k \, .
\] 

\end{proof}

We now use the previous lemmas to prove the bound on the competitive ratio of our algorithm.

\begin{lemma}
 If $\varGamma \ge 2 \sigma OPT$, then \textsc{MPC-Approx} does not fail.
\label{lem:fail}
\end{lemma}

\begin{proof}
Let $\bfx^f$ and ($\mathbf{y}^f$, $\mathbf{z}^f$) be the values for the primal and dual variables when at line~\ref{line:mpcfail} in the algorithm. $\bfx^f$ may be infeasible for the primal since the current job may not yet be assigned, however, ($\mathbf{y}^f$, $\mathbf{z}^f$) are feasible for the dual. Let $\bfx^*$ and ($\mathbf{y}^*$, $\mathbf{z}^*$) be the optimal solution. Then 

\begin{equation}
 OPT ~ = ~ \lambda(\bfx^*) ~ = ~ \sum_{i \in [m_c]} y_i^*
\label{eqn:optdualopt}
\end{equation}

\noindent where the last equality follows from LP strong duality. For convenience of notation, let   $\nu := \ln (e m) + \tl(\bfx^f)$. Since $\bfx$ is non-decreasing, $\tl(\bfx^f) = \max_l \tl(\bfx^l)$. Then by Lemmas~\ref{lem:dualsumz} and~\ref{lem:dualyz}, $\mathbf{z}^f/\nu$ and $\mathbf{y}^f \cdot \varGamma/(\sigma \nu)$ are feasible values for the dual variables. Thus the optimal dual value $\sum_i y_i^*$ is at least as large as $\sum_i y_i^f \cdot \varGamma / (\sigma \nu)$. From~(\ref{eqn:initopt}), $OPT \ge \max_{k,j} p_{kj}/d_1 \rho \kappa_1$. Hence if $\varGamma \ge OPT$, the condition for Corollary~\ref{cor:primaldual2} is satisfied. From~(\ref{eqn:optdualopt}) and Corollary~\ref{cor:primaldual2},

\begin{equation*}
 OPT ~ \ge ~ \sum_i y_i^f \cdot \varGamma / (\sigma \nu) ~ \ge ~ ( \est(\bfx^f) - \ln m -  \frac{OPT}{\varGamma}) \cdot \varGamma / (\sigma \nu) \, ,
\end{equation*}

\noindent or, rearranging terms,

\begin{equation*}
 \frac{\sigma \nu}{\varGamma} OPT + \frac{OPT}{\varGamma} + \ln m ~ \ge ~ \est(\bfx^f) \, .
\end{equation*}

\noindent Substituting the value of $\nu$, and since $\est(\bfx^f) \ge \tl(\bfx^f)$,

\begin{equation}
 \frac{OPT}{\varGamma} \sigma \left( \ln (em) + \tl(\bfx^f) \right) + \frac{OPT}{\varGamma} + \ln m ~ \ge ~ \tl(\bfx^f) \, .
 \label{eqn:lambdabound}
\end{equation}

\noindent Using the bound on $OPT$ from the statement of the lemma, and since $\sigma \ge 1$,

\[
 \frac{\ln (em) + \tl(\bfx^f)}{2} + \ln(em) ~ > ~ \tl(\bfx^f) \, ,
\]

\noindent and simplifying yields $\tl(\bfx^f) < 3 \ln (em)$. Hence, if $\varGamma \ge 2 \sigma OPT$, the algorithm does not fail.
\end{proof}

\begin{lemma}
If $4 \sigma OPT \ge \varGamma \ge \max_{k,j} p_{kj}/ (d_1 \rho \kappa_1)$ and \textsc{MPC-Approx} does not fail, it returns a $8 \sigma \ln (em)$-competitive solution.
\label{lem:main}
\end{lemma}

\begin{proof}
Since $\varGamma \ge \max_{k,j} p_{kj}/ (d_1 \rho \kappa_1)$, from~(\ref{eqn:lambdabound}),

\[
 OPT \sigma \left( \ln (em) + \tl(\bfx^f) \right) + OPT + \varGamma \ln m ~ \ge ~ \varGamma \tl(\bfx^f) ~ = ~ \lambda(\bfx^f) \, .
\]

\noindent Using the upper bound on $\varGamma$, and since $\tl(\bfx^f) \le 3 \ln (em)$ by Corollary~\ref{cor:init2},

\[
 4 OPT \sigma \ln (em) + OPT + 4 OPT \sigma \ln m ~ \ge ~ \lambda(\bfx^f) \, .
\]

This proves the lemma.
\end{proof}

Since $OPT \ge \max_{k,j} p_{kj}/d_1 \rho \kappa_1$, Lemmas~\ref{lem:fail} and~\ref{lem:main} imply

\begin{theorem}
 If $4 \sigma OPT \ge \varGamma \ge 2 \sigma OPT$, then \textsc{MPC-Approx} does not fail and returns a $8 \sigma \ln (em)$-competitive solution.
 \label{thm:main}
\end{theorem}

\subsection{Proceeding Without an Estimate on OPT.}
\label{sec:mpcdoubling}

We now discuss how to proceed without an estimate on OPT. We use a doubling procedure commonly used in online algorithms. We initially set $\varGamma$ $=  \max_{k,j:p_{kj} > 0} p_{kj} / (d_1 \rho \kappa_1)$ and use this value to scale the packing constraints. We run Algorithm \textsc{MPC-Approx} with the scaled values. If the algorithm fails, we double $\varGamma$, scale the packing constraints by the new value of $\varGamma$ and run the algorithm again. We repeat this each time the algorithm fails.

Each execution of Algorithm \textsc{MPC-Approx} is a \emph{trial}. Each trial $\tau$ has distinct primal and dual variables $(\lambda(\tau)$, $\bfx(\tau))$ and $(\mathbf{y}(\tau)$, $\mathbf{z}(\tau))$ that are initialized at the start of the trial and increase as the trial proceeds. At the start of the trial, each $x_j(\tau)$ is initialized to $x_j^0(\tau) = 1/(d_1^2 \rho \kappa_1)$. If a trial fails, we double the value of $\varGamma$ and proceed with the next trial with new primal and dual variables. Thus in every trial, $\varGamma \ge \max_{k,j:p_{kj} > 0} p_{kj} / (d_1 \rho \kappa_1)$. 

Our final value for $(\bfx, \lambda)$ is the sum of the values obtained in each trial. Thus, our variables are non-decreasing. Let $\varGamma(\tau)$ be the value of $\varGamma$ used in trial $\tau$, and $\lambda^f(\tau)$ be the value of the primal $\lambda(\tau)$ when trial $\tau$ ends. $T$ is the last trial, i.e., the algorithm does not fail in trial $T$. Since $\bfx$ obtained by the algorithm is the sum of $\bfx(\tau)$ in each trial $\tau$, the value of the primal objective obtained by the algorithm is at most $\sum_{\tau \le T} \lambda^f(\tau)$. Then

\begin{theorem}
 The value of the primal objective $\sum_{\tau \le T} \lambda^f(\tau)$ obtained is $O(\ln m \ln (d \rho \kappa)) OPT$.
 \label{thm:doubling}
 \end{theorem}
 
 We first show a bound on $\varGamma$ in any trial.
 
 \begin{lemma}
  In any trial, $\varGamma \le 4 \sigma OPT$.
  \label{lem:boundtl}
 \end{lemma}

\begin{proof}
 Initially, $\varGamma = \max_{k,j:p_{kj} > 0} p_{kj} / (d_1 \rho \kappa_1)$ $\le OPT$ by~(\ref{eqn:initopt}). Hence the lemma is true for the first trial. Since $\varGamma$ is doubled after each failed trial, by Lemma~\ref{lem:fail} some trial with $\varGamma \le 4 \sigma OPT$ will not fail. Hence, for every trial, $\varGamma \le 4 \sigma OPT$.
\end{proof}

\noindent \emph{Proof of Theorem~\ref{thm:doubling}.}
Define $\tl^f(\tau) := \lambda^f(\tau)/\varGamma(\tau)$. By Corollary~\ref{cor:init2}, $\tl(\tau) \le 3 \ln (em)$ at the start of any phase. Within a phase, each variable gets multiplied by at most a factor of $\mu = 1 + 1/(3 \ln (em))$. Hence when trial $\tau$ fails, $\tl^f(\tau) = \lambda^f(\tau)/\varGamma(\tau) \le 1 + 3 \ln (e m) \le 4 \ln(em)$, or $\lambda^f(\tau) \le 4 \varGamma(\tau) \ln (e m)$. Since the value of $\varGamma(\tau)$ doubles after each trial,

\begin{equation}
 \sum_{\tau \le T} \lambda^f(\tau) ~ \le ~ 4 \ln (e m) \sum_{\tau \le T} \varGamma(\tau)  ~ = ~ 4 \ln (e m) \sum_{\tau \le T} 2^{\tau-T} \varGamma(T) ~ \le ~ 8 \ln (e m) \varGamma(T) \, .  \label{eqn:trial1}
\end{equation}

\noindent Thus, from~(\ref{eqn:trial1}) and Lemma~\ref{lem:boundtl}, $\sum_{\tau \le T} \lambda^f(\tau) \le 32 \sigma \ln (em)  OPT$, proving the theorem. \qed
 
\subsection{A Lower Bound for Mixed Packing and Covering Online}
\label{sec:lower}

We give a lower bound on the competitive ratio of any deterministic algorithm for online mixed packing and covering. Given upper bounds $m$ and $d$ on the number of packing constraints and on the number of variables in any (packing or covering) constraint respectively, we give an example to show the following lower bound.

\begin{theorem}
 Any deterministic algorithm for OMPC is $\Omega(\log (d/ \log m) \log m)$-competitive.
 \label{thm:lower}
\end{theorem}

Our algorithm for OMPC in Section~\ref{sec:algoalgo} is thus nearly tight. For parameters $d$ and $m$, we give an example which has $m$ packing constraints, at most $2d$ variables in each covering constraint, and at most $d \log m$ variables in each packing constraint. For this example, we show that $OPT = 1$ and any deterministic algorithm gets value $\Omega(\log d \log m)$. The theorem follows.

We assume that both $d$ and $m$ are powers of 2 without loss of generality, otherwise we redefine $d$ to be the highest power of 2 that is at most the given value of $d$, and redefine $m$ similarly. Our example has $2(m-1)d$ variables. We partition the variables into $2(m-1)$ pairwise disjoint sets, with each set consisting of $d$ variables, and use $B_i$ to refer to the $i$th set. We refer to these sets as \emph{blocks}. For any set of variables $S$, we use $w(S)$ to refer to the sum of the values assigned by the algorithm to the variables in $S$, and use $\Sigma(S)$ to refer to the expression $\sum_{x \in S} x$.

We first show how given two blocks $B_1$ and $B_2$ of size $d$, we can construct covering constraints so that $w(B_i) \ge H_d/2$ for one of $i \in \{1,2\}$, while the constraints can be satisfied by setting $x_j = 1$ for a single variable $x_j \in B_{i'}$, $i' \neq i$. The covering constraints are given by Algorithm~\ref{algo:example}. $H_d$ refers to the $d$th harmonic number. 

\begin{algorithm}[h]
\caption{Given blocks $B_1$ and $B_2$:}
\label{algo:example}
\begin{algorithmic}[1]
 \STATE $B_1' \leftarrow B_1$, $B_2' \leftarrow B_2$
 \FOR {$i = 1 \to d-1$}
  \STATE Offer the covering constraint $\Sigma(B_1' \cup B_2') \ge 1$
  \STATE Let $x_1$, $x_2$ be the variables assigned maximum value in $B_1'$ and $B_2'$ respectively
  \STATE $B_1' \leftarrow B_1' \setminus \{x_1\}$, $B_2' \leftarrow B_2' \setminus \{x_2\}$
 \ENDFOR
 \STATE Offer the covering constraint $\Sigma(B_1' \cup B_2') \ge 1$
\end{algorithmic}
\end{algorithm}

\begin{lemma}
 Either $w(B_1) \ge H_d/2$ or $w(B_2) \ge H_d/2$.
 \label{lem:blockweight}
\end{lemma}

\begin{proof}
 Let $x_1$, $x_2$ be the variables assigned maximum value in $B_1'$ and $B_2'$ respectively in the $i$th iteration of the for loop. Since $|B_1'| = |B_2'| = (d+1-i)$, and $x_1$, $x_2$ have maximum value in $B_1'$ and $B_2'$ respectively, $x_1 \ge w(B_1')/(d+1-i)$ and $x_2 \ge w(B_2')/(d+1-i)$. Further, since $w(B_1') + w(B_2') \ge 1$, $x_1 + x_2 \ge 1/(d+1-i)$. Thus when all the covering constraints are satisfied, $w(B_1 \cup B_2) \ge H_d$, and hence either $w(B_1) \ge H_d/2$ or $w(B_2) \ge H_d/2$.
\end{proof}

Assume $w(B_1) \ge H_d/2$. Then there exists some variables $x_j \in B_2$ which is in each covering constraint introduced, and hence all the constraints can be satisfied by setting this variable to 1.

For the complete example, consider a complete binary tree with $m$ leaf nodes. Each node in this tree except the root corresponds to a block, and no two nodes correspond to the same block. Our packing constraints correspond to the leaf nodes, with packing constraint $k$ being $E(\cup_{i \in Q_k} B_i) \le \lambda$ where $Q_k$ is the set of blocks encountered on the path from the root to the leaf node corresponding to packing constraint $k$.

For a node $v$, let $l$ and $r$ be the left and right child respectively, and let $B_l$ and $B_r$ be the blocks corresponding to these children. We now start from the root node and walk to a leaf node in the following way. When we are at node $v$, we run Algorithm~\ref{algo:example} with blocks $B_l$ and $B_r$. If $w(B_l) \ge w(B_r)$ we step on the left child and ``mark'' the right child, else we step on the right child and ``mark'' the left child. We continue with the node we stepped on as node $v$, and continue in this manner until we reach a leaf node. Then say the leaf node we arrive at corresponds to packing constraint $k$. Since each block on the path from the root to this leaf node (except the root) has weight at least $H_d/2$ by Lemma~\ref{lem:blockweight}, and the path from the root to any leaf has $\log n$ nodes, the total value assigned by the algorithm to variables in this constraint is at least $\log m \cdot H_d/2$, thus $\lambda \ge \log m \cdot H_d/2$. 

On the other hand, setting a single variable to 1 in each marked node satisfies all the covering constraints. The path from the root to any leaf node contains at most one marked node, since when we mark a node, the blocks in the subtree rooted at that node do not appear in any covering constraint. Thus, we can satisfy the covering constraints by setting at most a single variable to 1 in each packing constraint, where for a packing constraint, the variable set is in the marked node (if any) in the path from the root to the leaf corresponding to the packing constraint. Hence OPT = 1, and any deterministic algorithm obtains $\lambda \ge \log m \cdot H_d/2$.

%% file: of-flcc2.tex
\section{UMSC and CCFL}
\label{sec:ccfl}

We now build upon the techniques in Section~\ref{sec:algo} and give a polylogarithmic-competitive integral algorithm for UMSC and CCFL. Recall that both UMSC and CCFL generalize online set-cover, for which there is a lower bound of $\Omega(\log m \log n)$ on the competitive ratio assuming BPP $\neq$ NP~\cite{AlonAG09}. Our algorithm is thus tight modulo a logarithmic factor.

CCFL generalizes UMSC; an instance of UMSC is an instance of CCFL where each facility corresponds to a machine and has unit capacity, each client corresponds to a job, and all assignment costs are zero. We describe an algorithm for CCFL, which also gives an algorithm for UMSC with the same competitive ratio. Further, in CCFL, the demand $p_{ij}$ and capacity $u_i$ only appear as the ratio $p_{ij}/u_i$. We redefine $p_{ij}$ as this ratio, and assume that the capacity of every facility is unity. The congestion of a facility is then the sum of the demands of clients assigned to the facility.

We use $[m]$ to denote the set of facilities, and $[n]$ to denote the set of clients. We exclude trivial instances and assume $m$, $n$ are at least 2. We assume that $n$ is given offline. Variables $i, i'$ index facilities, while $j, j'$ index clients. Clients appear in order of their indices: the first client is client 1, and the last client is client $n$. The \emph{total cost} $Z$ of an assignment of clients to facilities is the sum of the maximum congestion, fixed-charge, and assignment costs. $Z^*$ is the total cost of the optimal assignment. We will assume we are given an estimate $Z$ with $Z^* \le Z \le 2Z^*$. In the absence of this estimate, we use a doubling approach as described previously; Section~\ref{sec:ccflZdouble} explains how doubling can be used for this particular problem. For client $j$, $F_j(Z) := \{i: p_{ij} + a_{ij} + c_i \le Z\}$. Since assigning client $j$ to a facility $i$ not in $F_j(Z)$ would incur total cost larger than $Z$, we fix the fractional assigment of client $j$ to any facility $i$ not in $F_j(Z)$ to be zero. 

We first give an algorithm that obtains a fractional solution for the problem, and then use a randomized rounding technique adapted from~\cite{KhullerLS10} to obtain an integral solution. A fractional solution corresponds to a solution to linear program CCFL-LP1($Z$), which takes $Z$ as a parameter. A client may be fractionally assigned to facilities, and the sum of these fractional assignments for each client must be at least 1. $x_{ij}$ is the fractional assignment of client $j$ to facility $i$. Facilities can also be opened fractionally, and $y_i$ is the fraction to which facility $i$ is open. The fraction to which any facility is opened is an upper bound on both the fractional assignment of any client to that facility, and the ratio of congestion of the facility to $Z$. $\lambda$ is an upper bound on $y_i$ for each facility. Since for every facility the fraction $y_i$ is an upper bound on the ratio of congestion to $Z$, $Z \lambda$ is the maximum congestion of any facility.

\begin{table}[h]
 \begin{center}
 \begin{tabular}{c}
  $
   \begin{array}{rrl}
    \mbox{\textbf{CCFL-LP1($Z$):}} & \multicolumn{2}{l}{~  \min ~ \displaystyle \sum_{i \in [m]} c_i y_i + Z \lambda + \sum_{i \in [m], \, j \in [n]} a_{ij} x_{ij} \hspace*{0.6in}} \\
    &  \sum_{i \in [m]} x_{ij} & \ge ~ 1, \hspace{0.2in} \forall j \in [n] \\
    &  y_i - x_{ij} & \ge ~ 0, \hspace{0.2in} \forall i \in [m], j \in [n] \\
    &  Z y_i - \sum_{j \in [n]} p_{ij} x_{ij} & \ge ~ 0, \hspace{0.2in} \forall i \in [m] \\
    & \lambda - y_i & \ge ~ 0, \hspace{0.2in} \forall i \in [m] \\
    & \lambda & \ge ~ 1
   \end{array}
  $ 
  \end{tabular}
 \end{center}
\end{table}

$\bfx$ is the vector $(x_{ij})_{i \in [m],\, j \in [n]}$, and similarly $\bfy$ is the vector $(y_i)_{i \in [m]}$. $\optone(Z)$ is the cost of the optimal solution to CCFL-LP1($Z$). Since we restrict assignment of client $j$ to facilities in $F_j(Z)$, if $F_j(Z) = \emptyset$, CCFL-LP1($Z$) is infeasible. We assume that $Z^* \le Z \le 2 Z^*$, and hence CCFL-LP1($Z$) is feasible. Also, since $\lambda \ge 1$,

\begin{fact}
 $\optone(Z) \ge Z$. \label{fact:ccfloptone}
\end{fact}

Define $\rho := \max_j \frac{\max_i (c_i + p_{ij} + a_{ij})}{\min_i (c_i + p_{ij} + a_{ij})}$. We use techniques from Section~\ref{sec:algo} and obtain an $O(\log (mn) \log (mn \rho))$-competitive fractional solution for CCFL-LP1($Z$). We begin by highlighting the major differences between CCFL and OMPC, and briefly mention how they are dealt with. 

Firstly, the linear programm CCFL-LP1($Z$) no longer consists solely of packing and covering constraints, since there are variables with negative coefficients. However, we can still express the objective of CCFL-LP1($Z$) as a function of the vector $\bfx$: given $\bfx$ that satisfies the first set of constraints in CCFL-LP1($Z$), define $y_i(\bfx)$ as the minimum value of $y_i$ that satisfies the remaining constraints; $\lambda(\bfx)$ is defined correspondingly. We then proceed as in \textsc{MPC-Approx}: we define $\cost(\bfx)$ as a derivable approximation to the objective, and $\rate(\bfx)$ as the derivative of $\cost(\bfx)$. $\rate(\bfx)$ is then used to obtain the multiplicative updates for variables $\bfx$.

Secondly, whereas earlier $\rate(\bfx)$ was a continuous function of $\bfx$ for OMPC, this is no longer the case for CCFL. Now $\rate(\bfx)$ depends on whether the second set of constraints in CCFL-LP1($Z$) are tight. We deal with this by separating the updates where the second set of constraints is tight, and increment $\bfx$ differently in each case.

Thirdly, in Section~\ref{sec:algo} for each variable $x_j$ since the packing constraints are available offline, the coefficients of $x_j$ are also known offline. This allows us to initialize variables offline. In CCFL, clients arrive online, and when each client arrives, we learn its demand and assignment cost for each facility. Thus the coefficients $d_{ij}$ and $c_{ij}$ of each variable $x_{ij}$ in CCFL-LP1($Z$) are received online, and these variables need to be initialized online. We show that the increase in $\cost(\bfx)$ because of these initializations is small.

Fourthly, CCFL-LP1($Z$) is a parametric LP. If we do not know $Z^*$, we use a doubling procedure twice: once to obtain $Z$ such that $Z^* \le Z \le 2Z^*$, and once again to scale CCFL-LP1($Z$) by $\varGamma$, as in Section~\ref{sec:algo}.

\subsection{A Fractional Solution for CCFL}

We start by scaling the CCFL-LP1($Z$) by a parameter $\varGamma$ to obtain LP2($Z$, $\varGamma$) and its dual, D2($Z$, $\varGamma$). $\opttwo$($Z$, $\varGamma$) is the cost of the optimal solution to LP2($Z$, $\tl$). In the following analysis we will keep the dual variable $\mu = 0$ and exclude it from further discussion. 

\begin{table}[h]
 \begin{center}
 \begin{tabular}{l|l}
  $
   \begin{array}{rl}
    \multicolumn{2}{c}{\mbox{\textbf{LP2($Z$, $\varGamma$)}:}  \hspace{0.1in} \min ~ \displaystyle  \sum_{i} c_i \ty_i + Z \tl + \frac{1}{\varGamma}\sum_{i,j} a_{ij} x_{ij} \hfill} \\
    \displaystyle \sum_i x_{ij} & \ge ~ 1, ~ \forall j \\ 
    \vspace{0.05in} \displaystyle \ty_i - \frac{x_{ij}}{\varGamma} & \ge ~ 0, ~ \forall i, j \\
     \displaystyle Z \ty_i - \sum_j p_{ij} \frac{x_{ij}}{\varGamma} & \ge ~ 0, ~ \forall i \\
    \displaystyle \tl - \ty_i & \ge ~ 0, ~ \forall i \\
    \displaystyle \tl & \ge ~ 1
   \end{array}
  $
  &
  $
   \begin{array}{rl}
    \multicolumn{2}{c}{\mbox{\textbf{D2($Z$, $\varGamma$)}:} \hspace{0.1in} \displaystyle \max ~ \sum_j \alpha_j  + \frac{1}{\varGamma}\mu \hfill} \\  \vspace{0.25em}
    \displaystyle \varGamma \alpha_j - \beta_{ij} - p_{ij} \gamma_i - a_{ij} & \le 0, ~ \forall i,j \\
    \displaystyle \sum_j \beta_{ij} + Z \gamma_i - \delta_i & \le c_i, ~ \forall i \\
    \displaystyle \sum_i \delta_i + \mu & \le Z
   \end{array}
  $
 \end{tabular}
\end{center}
\end{table}

\noindent 

\begin{fact}
For any $Z$, $\varGamma > 0$, $(\bfx,\bfy,\lambda)$ is a feasible solution to CCFL-LP1($Z$) of cost $Z'$ iff $(\bfx,\bfy/\varGamma,\lambda/\varGamma)$ is a feasible solution to LP2($Z,\varGamma$) of cost $Z'/\varGamma$.
 \label{fact:ccfllp12}
\end{fact}

Given a vector $\bfx$, let $\tw_i(\bfx) := \max_j x_{ij}/\varGamma$, $~ \tv_i(\bfx) := \sum_j p_{ij}x_{ij}/(Z \varGamma)$. Define $\ty_i(\bfx) := \tv_i(\bfx) + \tw_i(\bfx)$, and $\tl(\bfx) := \max_i \ty_i(\bfx)$. If $\bfx$ satisfies $\sum_i x_{ij} \ge 1$ for all $j$, then $(\bfx, \bty(\bfx), \tl(\bfx))$ is feasible for LP2($Z$, $\varGamma$), where $\bty$ is the $m$-vector of values $(\ty_i(\bfx))_{i \in [m]}$. As in Section~\ref{sec:algo},

\begin{equation}
 \displaystyle \est(\bfx) ~ := ~ \ln \sum_i \left( e^{\sum_j \frac{p_{ij}x_{ij}}{Z\varGamma}} \right) + \ln \left( \sum_{i,j} e^{\frac{x_{ij}}{\varGamma}} \right) 
 \label{eqn:ccfldefnest}
\end{equation}

\noindent is an upper bound on $\tl(\bfx)$, and $\cost(\bfx) := Z \cdot \est(\bfx) + \sum_i c_i \ty_i(\bfx) + \sum_{i,j} a_{ij} x_{ij}/\varGamma$ is an upper bound on the cost of the solution $(\bfx, \bty(\bfx), \tl(\bfx))$. The rate of change of $\cost(\bfx)$ is

 \begin{equation}
  \rate_{ij}(\bfx) ~ := ~ \frac{\partial \cost(\bfx)}{\partial x_{ij}} ~ = ~ Z \frac{\partial \est(\bfx)}{\partial x_{ij}} + \frac{c_i}{\varGamma} \left( \frac{p_{ij}}{Z} + 1_{ij} \right) + \frac{a_{ij}}{\varGamma} \label{eqn:ccfldefnrate}
 \end{equation}

\noindent where $1_{ij} = 1$ if $x_{ij} = \varGamma \tw_i$, and 0 otherwise.

Our algorithm is given as \assign. Define

\begin{equation}
 x_{ij}^0 ~ := ~ \frac{1}{2m n} \frac{\min_{i'} (c_{i'} + p_{i'j} + a_{i'j})}{c_i + p_{ij} + a_{ij}} \, . \label{eqn:ccfldefnxij0}
\end{equation}

\noindent At the beginning of the algorithm, before any requests arrive, set $x_{ij} = 0$ for all $i$, $j$. When client $j$ arrives, we initialize $x_{ij} = x_{ij}^0$ for all $i$. As long as $j$ is not fully assigned, we increment the fractional assignment $x_{ij}$ for each client $i$. The increment occurs in \emph{phases} where a phase is a single iteration of the while loop in the algorithm. The phases are indexed by $l$. $L$ is the set of all phase indices, and $L_j$ is the set of phase indices for phases executed to assign client $j$.


The increase in $x_{ij}$ in each phase is inversely proportional to $\rate_{ij}(\bfx)$. Define $\mu := 1 + \frac{1}{6 \ln(emn)}$. We also scale each update for client $j$ by $\epsilon_j(\bfx)$, where 

\begin{equation}
 \epsilon_j(\bfx) := (\mu - 1) \min_i \rate_{ij}(\bfx) \, .
 \label{eqn:ccfldefnepsilon}
\end{equation}
 
 \noindent This definition ensures that in each phase, the factor by which each variable is incremented is at most $\mu$. 

\begin{algorithm}[h]
\caption{\assign: When client $j$ arrives}
\label{algo:assign}
\begin{algorithmic}[1]
  \STATE $\forall i$, $x_{ij} \leftarrow x_{ij}^0$ \label{line:ccflalgoinit}
  \WHILE{$\sum_i x_{ij} < 1$}
    \STATE $l \leftarrow l + 1$, $\bfx^l \leftarrow \bfx$
    \FOR{$i \in F_j(Z)$}
      \IF {$\tw_i(\bfx^l) > x_{ij}/\varGamma$} \STATE $x_{ij} \leftarrow \min \left\{ \varGamma w_i(\bfx^l), x_{ij} \left(1 + \frac{\epsilon_j(\bfx^l)}{\rate_{ij}(\bfx^l)} \right) \right\}$  \hspace {0.5in} /* see definitions in~(\ref{eqn:ccfldefnrate}),~(\ref{eqn:ccfldefnepsilon}) */ \ELSE \STATE $x_{ij} \leftarrow x_{ij} \left(1 + \frac{\epsilon_j(\bfx^l)}{\rate_{ij}(\bfx^l)} \right)$ \ENDIF
    \ENDFOR
  \STATE $\alpha_j \leftarrow \alpha_j + e \epsilon_j(\bfx^l)$ \hspace{2.275in} /* for analysis */
\STATE \textbf{if} $\cost(\bfx) > 5 Z \ln (emn)$ \textbf{then return} FAIL \label{line:ccflalgofail}
\ENDWHILE
\end{algorithmic}
\end{algorithm}

\begin{fact}
 At any stage in the algorithm, $x_{ij} \le \mu$ for all $i$, $j$.
 \label{fact:ccflmu}
\end{fact}

Our analysis of \assign follows the primal-dual analysis in Section~\ref{sec:algo} closely. One difference between the current problem and Section~\ref{sec:algo} is that since we receive requests online, we do not know the coefficients of variables in the packing constraints offline, and unlike Section~\ref{sec:algo} cannot initialize our variables offline. In \assign as each request is received, we obtain the corresponding coefficients, and initialize our variables $x_{ij} = x_{ij}^0$ in line~\ref{line:ccflalgoinit}. For client $j$, define $\init_j$ as the change in $\cost(\bfx)$ on execution of line~\ref{line:ccflalgoinit} when $j$ arrives. $\cost(\bfx)$ is initially $\cost(\mathbf{0})$ and increases either due to line~\ref{line:ccflalgoinit} or within a phase. We begin our analysis by showing bounds on the change in $\cost(\bfx)$ due to these. 

\begin{lemma}
 If $\varGamma \ge 1$, then $\init_j \le Z / n$ for any client $j$.
 \label{lem:ccflinitbound}
\end{lemma}

\begin{proof}
 Fix a client $j$. Let $\bfx'$ and $\bfx''$ be the values of $\bfx$ before and after the execution of line~\ref{line:ccflalgoinit} on arrival of client $j$. We consider the differences $Z(\est(\bfx'') - \est(\bfx'))$ and $\sum_i c_i (\tv_i(\bfx'') + \tw_i(\bfx''))$ $- \sum_i c_i (\tv_i(\bfx') + \tw_i(\bfx'))$ $+ \sum_{i,j'} a_{ij'} (x_{ij'}'' - x_{ij'}')/\varGamma$ separately. Since $\cost(\bfx) = \sum_i c_i (\tv_i(\bfx) + \tw_i(\bfx)) + Z \est(\bfx) + \frac{1}{\varGamma} \sum_{i,j} a_{ij} x_{ij}$, the sum of these differences will give us $\init_j$.
 
 By definition of $\est(\bfx)$,
 
 \begin{align*}
 \est(\bfx'') - \est(\bfx') & ~ = ~ \ln  \sum_i \exp \left( \frac{\sum_{j'} p_{ij'} x_{ij'}''}{Z \varGamma} \right) - \ln  \sum_i \exp \left( \frac{\sum_{j'} p_{ij'} x_{ij'}'}{Z \varGamma} \right) \\
  & ~ + \ln \sum_{i, j'} \exp \left( \frac{x_{ij'}''}{\varGamma} \right) - \ln \sum_{i, j'} \exp \left( \frac{x_{ij'}'}{\varGamma} \right)\, .
 \end{align*}

 $\bfx''$ and $\bfx'$ differ only in values for client $j$, and are identical for other requests $j' \neq j$. Further, for all $i$ and $j' > j$, $x_{ij'}'' = x_{ij'}' = 0$. Thus
 
 \begin{align*}
 \est(\bfx'') - \est(\bfx') & ~ = ~ \ln \sum_i \exp \left(\frac{\sum_{j'<j} p_{ij'} x_{ij'}' + p_{ij} x_{ij}^0}{Z \varGamma} \right) - \ln \sum_i \exp \left(\frac{\sum_{j'<j} p_{ij'} x_{ij'}'}{Z \varGamma} \right) \\
  & ~ + \ln \sum_i \left( \sum_{j'<j} \exp \frac{ x_{ij'}'}{\varGamma}  + \exp \frac{x_{ij}^0}{\varGamma} + (n-j-1) \right) - \ln \sum_i \left( \sum_{j'<j} \exp \frac{ x_{ij'}'}{\varGamma}  + (n-j) \right)\, .
 \end{align*}
 
 Since $\varGamma \ge 1$ by assumption, $x_{ij}^0/\varGamma \le 1/(2mn)$. Further, since $p_{ij} \le Z$ for $i \in F_j(Z)$, $p_{ij} x_{ij}^0/(Z \varGamma)$ $\le x_{ij}^0/\varGamma$ $\le 1/(2mn)$. Substituting in the previous expression for $\est(\bfx'') - \est(\bfx')$ yields

 \begin{align*}
 \est(\bfx'') - \est(\bfx') & ~ \le ~ \ln \left( \exp \frac{1}{2mn} \sum_i \exp \frac{\sum_{j'<j} p_{ij'} x_{ij'}'}{Z \varGamma} \right) - \ln \sum_i \exp \left(\frac{\sum_{j'<j} p_{ij'} x_{ij'}'}{Z \varGamma} \right) \\
 & ~ + \ln \frac{\sum_i \left( \sum_{j'<j} \exp \frac{ x_{ij'}'}{\varGamma}  + \exp \frac{1}{2mn} + (n-j-1) \right)}{\sum_i \left( \sum_{j'<j} \exp \frac{ x_{ij'}'}{\varGamma} + (n-j) \right)} \\
 & ~ = ~ \frac{1}{2mn} + \ln \left( 1+  \frac{\sum_i (\exp \frac{1}{2mn} - 1)}{\sum_i \left( \sum_{j'<j} \exp \frac{ x_{ij'}'}{\varGamma} + (n-j) \right)} \right)\\
 & ~ \le ~ \frac{1}{2mn} + \ln \left( 1 + \frac{\sum_i (\exp \frac{1}{2mn} - 1)}{\sum_i n} \right) ~ \le ~ \frac{1}{2mn}  + \frac{\exp \frac{1}{2mn}-1}{n} \,
 \end{align*}

\noindent where the last inequality is because $1+a \le \exp(a)$ for all $a \in \mathbb{R}$. Using $m,n \ge 2$, 

\begin{align}
 Z (\est(\bfx'') - \est(\bfx')) & ~ \le ~ \frac{Z}{4n} + \frac{Z \exp (\frac{1}{8})-1}{n} ~ \le ~ \frac{Z}{2n} \, . \label{eqn:ccflestbound1}
\end{align}

For the second bound, since $x_{ik}'' = x_{ik}'$ for $k \neq j$, $\tw_i(\bfx'') - \tw_i(\bfx') \le x_{ij}''/\varGamma$ $= x_{ij}^0/\varGamma$. Further since $p_{ij} \le Z$ for $i \in F_j(Z)$, $\tv_i(\bfx'') - \tv_i(\bfx')$ $=\frac{p_{ij}x_{ij}^0}{Z \varGamma}$ $\le x_{ij}^0/\varGamma$. Hence
  
  \begin{align}
   \sum_i \left(c_i (\tv_i(\bfx'') + \tw_i(\bfx'') - \tv_i(\bfx') - \tw_i(\bfx')) + \sum_{j'} \frac{a_{ij'}}{\varGamma} (x_{ij'}'' - x_{ij'}') \right) & ~ \le ~ \sum_i \frac{x_{ij}^0}{\varGamma} (2 c_i + a_{ij}) \nonumber \\
   & ~ \le ~ 2Z \frac{x_{ij}^0}{\varGamma} \nonumber \\
   & ~ \le ~ \frac{Z}{2n} \label{eqn:ccflestbound2}
  \end{align}

\noindent where the second inequality is because $c_i + a_{ij} \le Z$ for $i \in F_j(Z)$, and the last inequality is by definition of $x_{ij}^0$ and since $\varGamma \ge 1$ and $m \ge 2$. From~(\ref{eqn:ccflestbound1}) and~(\ref{eqn:ccflestbound2}), the lemma follows.
\end{proof}

\begin{corollary}
 If $\varGamma \ge 1$, then at the beginning of every phase, $\cost(\bfx) \le 6Z \ln (emn)$.
 \label{cor:ccflphasebegin}
\end{corollary}

\begin{proof}
 When the previous phase ended, since \assign did not fail, $\cost(\bfx) \le 5 Z\ln (emn)$. Between phases, $\bfx$ is incremented by at most the arrival of a client, which increases $\cost(\bfx)$ by at most $Z/n$ by Lemma~\ref{lem:ccflinitbound}.
\end{proof}

We now show that the increase in $\cost(\bfx)$ in any phase of the algorithm is bounded from above by the change in the dual objective.

\begin{lemma}
 If $\varGamma \ge 1$, the increase in $\sum_j \alpha_j$ is an upper bound on the increase in $\cost(\bfx)$ in every phase.
 \label{lem:ccflincr}
\end{lemma}

\begin{proof}
 We show the lemma for phase $l$, corresponding to request $r$. Let $\cost^l$ and $\cost^{l+1}$ be the values of $\cost(\bfx)$ before and after the variables are incremented in phase $l$ respectively. We will show that $\cost^{l+1} - \cost^l$ $\le e \, \epsilon_j(\bfx^l)$. Since $e \, \epsilon_j(\bfx^l)$ is the increase in the dual objective, this will prove the lemma.
 
 Our proof follows the proof for Lemma~\ref{lem:primaldual}. Let $\bfx^l$ and $\bfx^{l+1}$ be the values of $\bfx$ before and after being incremented in phase $l$. In phase $l$, only the variables correpsonding to client $j$ get incremented. We ignore variables for the other requests, and for each $x_{ij}$, let $g_i(t) := x_{ij}^l + (x_{ij}^{l+1} - x_{ij}^l)t$ be defined for $0 \le t \le 1$. Note that $g_i(0) = x_{ij}^l$ and $g_i(1) = x_{ij}^{l+1}$. Define $\mathbf{g}(t) = (g_i(t))_{i \in F_j(Z)}$. With some abuse of notation, any function of $\bfx$, say $h(\bfx)$, can be written as a function of $t$, with $h(t) := h(\mathbf{g}(t))$. Thus, the functions $\cost(\bfx)$, $\rate_{ij}(\bfx)$, $\ty_i(\bfx)$ and $\tl(\bfx)$ can be written as functions of $t$. In particular, $\rate_{ij}(t) := d \, \cost(t) / d g_i(t)$. 
 
 We use these alternate expressions in the remainder of the proof. By the chain rule,
 
 \[
  \frac{d \, \cost(t)}{dt} ~ = ~ \sum_i \frac{\partial \cost(t)}{\partial g_i(t)} \frac{d g_i(t)}{dt} ~ = ~ \sum_i \rate_{ij}(t) \frac{d g_i(t)}{dt} \, ,
 \]

\noindent and hence,

\begin{equation}
 \cost^{l+1} - \cost^l ~ = ~ \int_{t=0}^1 \frac{d \cost(t)}{dt} dt ~ = ~ \sum_i \int_{t=0}^1 \rate_{ij}(t) \frac{d g_i(t)}{dt} dt \, . \label{eqn:ccflrate1}
\end{equation}

The function $\rate_{ij}(t)$ may not be continuous if for some $i$, $\varGamma \tw_i(\bfx^{l+1}) = x_{ij}^{l+1}$, but $\varGamma \tw_i(\bfx^l) > x_{ij}^l$. By our choice of updates in \assign, the discontinuity is only at the point $t=1$; hence we redefine $\rate_{ij}(1) := \lim_{t \rightarrow 1^-} \rate_{ij}(t)$. Since we change $\rate_{ij}(t)$ at a single point,~(\ref{eqn:ccflrate1}) is still true. By Corollary~\ref{cor:ccflphasebegin}, $\cost^l \le 6Z \ln (emn)$. Also, each variable gets incremented by at most a factor of $\mu$ in a phase. Then by Lemma~\ref{lem:ccfltechrate}, $\rate_{ij}(t) \le e \, \rate_{ij}(0)$ for $0 \le t \le 1$, hence

\[
 \cost^{l+1} - \cost^l ~ \le ~ e \sum_i \rate_{ij}(\bfx^l) \int_{t=0}^1 \frac{d g_i(t)}{dt} dt ~ = ~ e \sum_i \rate_{ij}(\bfx^l) ( x_{ij}^{l+1} - x_{ij}^l) \, .
\]

\noindent Since in phase $l$ each variable is multiplied by $1 + \frac{\epsilon_j(\bfx^l)}{\rate_{ij}(\bfx^l}$, 

\[
 \cost^{l+1} - \cost^l ~ \le ~ e \sum_i \rate_{ij}(\bfx^l) \frac{\epsilon_j(\bfx^l) x_{ij}^l}{\rate_{ij}(\bfx^l)} ~ = ~ e \sum_i \epsilon_j(\bfx^l) x_{ij}^l ~ \le ~ e \epsilon_j(\bfx^l)
\]

\noindent where the last inequality follows since, on entering the for loop, $\sum_i x_{ij}^l < 1$. Since $e \epsilon_j(\bfx^l)$ is the increase in the dual objective, this proves the lemma.
\end{proof}

We now discuss our dual variables and show feasibility. By definition,

\begin{equation}
\displaystyle \frac{\partial \est(\bfx)}{\partial x_{ij}} ~ = ~ \frac{1}{\varGamma} \left[ \frac{p_{ij}}{Z} \frac{\exp(\sum_{j'} p_{ij'} x_{ij'}/(Z \varGamma))}{\sum_{i'} \exp(\sum_{j'} p_{i'j'} x_{i'j'}/(Z \varGamma))} + \frac{\exp(x_{ij}/\varGamma)}{\sum_{i', j'} \exp(x_{i'j'}/\varGamma)} \right] \, . 
\label{eqn:ccflest}
\end{equation}

We use the following notation. Recall that for a client $j$, $F_j(Z)$ is the set of facilities $i$ with $c_i + p_{ij} + a_{ij} \le Z$, and $x_{ij} = 0$ for any $i \not \in F_j(Z)$. For all $j$ and $i \not \in F_j(Z)$, we leave the corresponding dual variables undefined. For facility $i$, $n_i$ is the index of the first client $j$ so that $i \in F_j(Z)$. 

\begin{eqnarray}
 \forall i, ~ z_i(n_i-1)  ~ := ~ x_{in_i}^0 \, , ~ \mbox { and } \forall j \ge n_i, ~ z_i(j)  ~ := ~ \max_{j' \le j} x_{ij'} \, ,
 \label{eqn:ccfldefnzi}
 \end{eqnarray}
 
 \vspace{-0.2in}
 
 \begin{eqnarray}
 \chi_{ij} & \displaystyle ~ := ~ \max_{l \in L_j} \frac{\exp(x_{ij}^l/\varGamma)}{\sum_{i', j'} \exp(x_{i'j'}^l/\varGamma)} \hspace*{0.6in} & ~ \forall i \in [m], ~ \forall j: i \in F_j(Z) \, , \hfill \label{eqn:ccflchi} \\
 \eta_i & \displaystyle ~ := ~ \max_{l \in L} \frac{\exp(\sum_{j'} p_{ij'} x_{ij'}/(Z \varGamma))}{\sum_{i'} \exp(\sum_{j'} p_{i'j'} x_{i'j'}/(Z \varGamma))} \hfill & ~ \forall i \in [m] \, . \hfill \nonumber
\end{eqnarray}

From~(\ref{eqn:ccflest}) and these definitions, for any $i$, $j$ and any phase $l \in L_j$,

\begin{equation}
 \frac{\partial \est(\bfx^l)}{\partial x_{ij}} ~ \le ~ \frac{1}{\varGamma} \left( \frac{p_{ij}}{Z} \eta_i + \chi_{ij} \right) \, .
 \label{eqn:ccflestbound}
\end{equation}

Define $\sigma := 4e^2 \ln (2 \mu m n \rho)$. Set the dual variables:

\begin{align}
 \beta_{ij} & ~ = ~ Z \chi_{ij} \frac{\sigma}{2e^2} + c_i \ln \frac{z_i(j)}{z_i(j-1)} \, , & ~ \forall i \in [m], ~ \forall j:i \in F_j(Z) \, , \nonumber \\
 \gamma_i & ~ = ~ \left( \eta_i + \frac{c_i}{Z} \right) \frac{\sigma}{2e^2} \, , & ~ \forall i \in [m] \, , \label{eqn:ccflduals} \\
 \delta_i & ~ = ~ Z\left( \sum_{j: i \in F_j(Z)} \chi_{ij} + \eta_i \right) \frac{\sigma}{2e^2} \, , & ~ \forall i \in [m] \, . \nonumber 
\end{align}

We define $z_i(j)$ for $j \in \{n_i-1$, $...$, $n\}$ since the definition of $\beta_{in_i}$ requires $z_i(n_i-1)$. We now show bounds on the infeasibility of each dual constraint in D2($Z$, $\varGamma$) in the following sequence of lemmas. 

\begin{lemma}
 For all $i$, $j$, $\alpha_j  \le \frac{e^2}{\varGamma} \left( \beta_{ij} + p_{ij}\gamma_i + a_{ij} \frac{\sigma}{2e^2} \right)$.
 \label{lem:ccflalpha}
\end{lemma}

\begin{proof}
 For any $i$, $j$, let $L_{ij}^>$ be the set of phases where $\tw_i > x_{ij}/\varGamma$ before $x_{ij}$ is incremented, and $L_{ij}^=$ is the set of phases where $\tw_i = x_{ij}/\varGamma$ before $x_{ij}$ is incremented. Then $L_{ij}^> \cup L_{ij}^= = L_j$. Let $l$ be the first phase executed when $x_{ij}/\varGamma = \tw_i$ at the beginning of the phase. Then in every subsequent phase in $L_{j}$, $x_{ij}/\varGamma = \tw_i$ at the beginning of the phase. Hence any phase $l \in L_{ij}^=$ occurs after all the phases in $L_{ij}^>$.
 
 In any phase in $L_{ij}^>$ except the last, $x_{ij}$ is incremented by $(1 + \epsilon_j(\bfx^l)/\rate_{ij}(\bfx^l))$. Before the last phase, $x_{ij} \le 1$. In the last phase, $x_{ij}$ is incremented by at most $(1 + \epsilon_j(\bfx^l)/\rate_{ij}(\bfx^l))$ $\le \mu$. Hence
 
 \[
  x_{ij}^0 \prod_{l \in L_{ij}^>} \left( 1 + \frac{\epsilon_j(\bfx^l)}{\rate_{ij}(\bfx^l)} \right) ~ \le ~ \mu \, .
 \]

\noindent Using $1+a \ge e^{a/e}$ for $0 \le a \le 1$, and by rearranging the terms,

\[
 \exp \left(\sum_{l \in L_{ij}^>} \frac{\epsilon_j(\bfx^l)}{e \, \rate_{ij}(\bfx^l)} \right) ~ \le ~ \frac{\mu}{x_{ij}^0} \, .
\]

\noindent Taking the natural log on both sides, and observing that $x_{ij}^0 \ge 1/(2 m n \rho)$,

\[
 \left(\sum_{l \in L_{ij}^>} \frac{\epsilon_j(\bfx^l)}{e \, \rate_{ij}(\bfx^l)} \right) ~ \le ~ \ln \frac{\mu}{x_{ij}^0} ~ \le ~ \ln (2 \mu m n \rho) \, ,
\]

\noindent and hence

\begin{equation}
 \sum_{l \in L_{ij}^>} \epsilon_j(\bfx^l) ~ \le ~ e \ln (2 \mu m n \rho) \max_{l \in L_{ij}^>} \rate_{ij}(\bfx^l) \, .
 \label{eqn:schedalpha1}
\end{equation}

If $L_{ij}^= \neq \emptyset$, then during the execution of phases in $L_{ij}^=$, $x_{ij}$ increases from an initial value of $z_i(j-1)$ to $z_i(j)$ after the completion of the phases in $L_{ij}^=$. If $j$ is the first client in $C_i(Z)$, then its initial value is $x_{ij}^0$, hence $z_i(j-1) = x_{ij}^0$ as defined in~(\ref{eqn:ccfldefnzi}). By a similar analysis as for $L_{ij}^>$,

\begin{equation}
 \sum_{l \in L_{ij}^=} \epsilon_j^l ~ \le ~ e \left( \ln \frac{z_i(j)}{z_i(j-1)} \right)  \max_{l \in L_{ij}^=} \rate_{ij}(\bfx^l)  \, .
 \label{eqn:schedalpha2}
\end{equation}

The dual variable $\alpha_j$ gets incremented by $e \epsilon_j(\bfx^l)$ in each phase for client $j$. Hence, $\alpha_j$ $= e \left( \sum_{l \in L_{ij}^>} \epsilon_j(\bfx^l) \right.$  $ \left. + \sum_{l \in L_{ij}^=} \epsilon_j(\bfx^l) \right)$. Thus, from~(\ref{eqn:schedalpha1}) and~(\ref{eqn:schedalpha2}),

\[
 \alpha_j ~ \le ~ e^2 \left(\ln (2 \mu m n \rho) \max_{l \in L_{ij}^>} \rate_{ij}(\bfx^l)   +  \ln \frac{z_i(j)}{z_i(j-1)} \max_{l \in L_{ij}^=} \rate_{ij}(\bfx^l)  \right) \, .
\]

\noindent Replacing the values of $\rate_{ij}(\bfx^l)$ for $l \in L_{ij}^>$ and $l \in L_{ij}^=$ from~(\ref{eqn:ccfldefnrate}),

\begin{align*}
 \alpha_j & ~ \le ~ e^2 \left(  Z \max_{l \in L_{ij}^>} \frac{\partial \est(\bfx^l)}{\partial x_{ij}} + \frac{c_i p_{ij}}{Z \varGamma} + \frac{a_{ij}}{\varGamma} \right) \ln (2 \mu m n \rho) \\
 & \qquad{} + \left( Z \max_{l \in L_{ij}^=} \frac{\partial \est(\bfx^l)}{\partial x_{ij}} + \frac{c_i p_{ij}}{Z \varGamma} + \frac{c_i}{\varGamma} + \frac{a_{ij}}{\varGamma} \right) \ln \frac{z_i(j)}{z_i(j-1)} \, .
 \end{align*}

\noindent For any client $j$, $x_{ij} \le \mu$, and hence $z_i(j) \le \mu$. Since $x_{ij} \ge x_{ij}^0$, $z_i(j) \ge x_{ij}^0$. Thus, $z_i(j)/z_i(j-1) \le \mu/x_{ij}^0 \le 2 \mu m n \rho$. Thus

\[
 \alpha_j ~ \le ~ 2 e^2 \left( Z \max_{l \in L_j} \frac{\partial \est(\bfx^l)}{\partial x_{ij}} + c_i \frac{p_{ij}}{Z \varGamma} + \frac{a_{ij}}{\varGamma} \right) \ln (2 \mu m n \rho)  + e^2 \frac{c_i}{\varGamma} \ln \frac{z_i(j)}{z_i(j-1)}   \, ,
\]

\noindent and from the expression for $\partial \est(\bfx^l)/\partial x_{ij}$ in~(\ref{eqn:ccflest}),

\begin{align}
 \alpha_j & ~ \le ~ 2 e^2 \left( \frac{Z}{\varGamma} \left(\frac{p_{ij}}{Z} \eta_i + \chi_{ij}\right) + c_i \frac{p_{ij}}{Z \varGamma} + \frac{a_{ij}}{\varGamma} \right) \ln (2 \mu m n \rho) + e^2 \frac{c_i}{\varGamma} \ln \frac{z_i(j)}{z_i(j-1)} \nonumber \\
  & ~ = ~ \frac{e^2}{\varGamma} \left[ \left(p_{ij} \left( \eta_i + \frac{c_i}{Z} \right) + Z \chi_{ij} + a_{ij} \right) \frac{\sigma}{2e^2} + c_i \ln \frac{z_i(j)}{z_i(j-1)} \right] \label{eqn:alphafinal} \, .
\end{align}

\noindent By definition, $\beta_{ij} = Z \chi_{ij} \sigma/(2e^2) + c_i \ln \frac{z_i(j)}{z_i(j-1)}$, and $\gamma_i = \left( \eta_i + \frac{c_i}{Z} \right) \sigma/(2e^2)$. Replacing these expressions in~(\ref{eqn:alphafinal}) yields $\alpha_j \le \frac{e^2}{\varGamma}\left(\beta_{ij} + p_{ij} \gamma_i + a_{ij} \sigma/(2e^2) \right)$, proving the lemma.
 \end{proof}

\begin{lemma}
 For all $i$, $\displaystyle \sum_{j: i \in F_j(Z)} \beta_{ij} + Z \gamma_i - \delta_i ~ \le ~ c_i \frac{\sigma}{e^2}$.
 \label{lem:ccflbeta}
\end{lemma}

\begin{proof}
 By definition~(\ref{eqn:ccflduals}), for any $i$,
 
 \[
  \sum_{j:i \in F_j(Z)} \beta_{ij} + Z \gamma_i - \delta_i ~ = ~ c_i \frac{\sigma}{2e^2} + c_i \sum_{j:i \in F_j(Z)} \ln \frac{z_i(j)}{z_i(j-1)} ~ = ~ c_i \frac{\sigma}{2e^2} + c_i \ln \frac{z_i(n)}{z_i(n_i-1)}
 \]

\noindent where $n$ is the total number of clients, and $n_i$ is the index of the first client $j$ such that $i \in F_j(Z)$. By definition~(\ref{eqn:ccfldefnzi}), $z(n_i-1) \ge 1/(2 m n \rho)$, and $z_i(n) \le x_{ij} \le \mu$ by Fact~\ref{fact:ccflmu}. Hence

 \[
  \sum_{j:i \in F_j(Z)} \beta_{ij} + Z \gamma_i - \delta_i ~ \le ~ 3 c_i \ln (2 \mu m n \rho) ~ \le ~ c_i \frac{\sigma}{e^2} \, .
 \]
\end{proof}

\begin{lemma}
 $\displaystyle \sum_i \delta_i ~ \le ~ Z(1 + \ln (mn) + \max_{l} \tl(\bfx^l)) \frac{\sigma}{e^2}$.
 \label{lem:ccfldelta}
\end{lemma}

\begin{proof}
 We show that 

\begin{equation}
\sum_{j,i \in F_j(Z)} \chi_{ij} ~ \le ~ 1 + \ln(mn) + \max_l \tl(\bfx^l)
\label{eqn:ccfldelta1}
\end{equation}

\noindent and

\begin{equation}
\sum_i \eta_i ~ \le ~ 1 + \ln(m) + \max_l \tl(\bfx^l)
\label{eqn:ccfldelta2}
\end{equation}

\noindent which suffices to prove the lemma. For~(\ref{eqn:ccfldelta1}), let $j \in [n]$ and $i \in F_j(Z)$, and define $\phi(i,j)$ as the phase that maximizes $\frac{e^{x_{ij}^l/\varGamma}}{\sum_{i', j'} e^{x_{i'j'}^l/\varGamma}}$. Define $b_{ij} :=e^{x_{ij}^{\phi(i,j)}/\varGamma}$. Since variables are nondecreasing, 

\[
\sum_{j,i \in F_j(Z)} \max_{l \in L_j}\frac{e^{x_{ij}^l/\varGamma}}{\sum_{i', j'} e^{x_{i'j'}^l/\varGamma}}  ~ \le ~ \sum_{j,i \in F_j(Z)} \frac{b_{ij}}{\sum_{(i',j'): \phi(i',j') \le \phi(i,j)} b_{i'j'}}
\]

\noindent and by Lemma~\ref{lem:genineq2}, this is at most $1 + \ln (\sum_{j,i \in F_j(Z)} b_{ij})$ $ \le 1 + \ln (mn) + \max_l \tl(\bfx^l)$.

Similarly, for~(\ref{eqn:ccfldelta2}), define $\phi(i)$ as the phase which maximizes $\frac{e^{\sum_j p_{ij} x_{ij}^l/(Z \varGamma)}}{\sum_{i'} e^{\sum_j p_{i'j} x_{i'j}/(Z \varGamma)}}$, and define $b_i$ $:= \exp \left(\sum_j p_{ij} x_{ij}^{\phi(i)}/(Z \varGamma)\right)$.

\[
\sum_i \max_{l \in L} \frac{e^{\sum_j p_{ij} x_{ij}^l/(Z \varGamma)}}{\sum_{i'} e^{\sum_j p_{i'j} x_{i'j}^l/(Z \varGamma)}} ~ \le ~ \sum_i \frac{b_i}{\sum_{i': \phi(i') \le \phi(i)} b_{i'}}
\]

\noindent and by Lemma~\ref{lem:genineq2}, this is at most $1 + \ln m + \max_l \tl(\bfx^l)$.
\end{proof}

Let $\nu := (1 + \ln (mn) + \max_l \tl(\bfx^l))$. Then 

\begin{lemma}
 For $\bfx$ obtained by \assign, the vectors $\mathbf{\alpha'} = \mathbf{\alpha}/(\nu \sigma)$, $\mathbf{\beta'} = e^2 \mathbf{\beta}/(\nu \sigma)$, $\mathbf{\gamma'} =  e^2 \mathbf{\gamma}/(\nu \sigma)$ and $\mathbf{\delta'} =  e^2 \mathbf{\delta}/(\nu \sigma)$ are feasible for D2($Z$, $\varGamma$).
 \label{lem:ccflduals}
\end{lemma}

\begin{proof}
 We show that the constraints in D2($Z$, $\varGamma$) are satisfied by $(\mathbf{\alpha'}, \mathbf{\beta'}, \mathbf{\gamma'}, \mathbf{\delta'})$. For the third constraint in D2($Z$, $\varGamma$), from Lemma~\ref{lem:ccfldelta}, $\sum_i \delta_i' = \sum_i \delta_i e^2 / (\sigma \nu) \le Z$
 
 For the second constraint, 

\[
 \sum_j \beta_{ij}' + Z \gamma_i' - \delta_i' ~ =  ~ \left(\sum_j \beta_{ij} + Z \gamma_i - \delta_i\right)e^2/(\nu \sigma) ~ \le ~ \left(\sum_j \beta_{ij} + Z\gamma_i - \delta_i \right) e^2/\sigma ~ \le ~ c_i
\]

\noindent where the first inequality is because $\nu \ge 1$, and the last inequality is from Lemma~\ref{lem:ccflbeta}.

For the first constraint, 

\[
 \alpha_j' - \frac{1}{\varGamma}(\beta_{ij}' + p_{ij} \gamma_i' + a_{ij}) ~ = ~ \frac{\alpha_j}{\nu \sigma} - \frac{e^2}{\varGamma} \left( \frac{\beta_{ij}}{\nu \sigma} + p_{ij} \frac{\gamma_i}{\nu \sigma}\right) - \frac{a_{ij}}{\varGamma} ~ \le ~ 0 \, ,
\]

\noindent where the last inequality follows from Lemma~\ref{lem:ccflalpha} and since $\nu \ge 1/2$. Hence, $(\alpha', \beta', \gamma', \delta')$ are feasible for D2($Z$, $\varGamma$).
\end{proof}

We now use the primal-dual framework to show that given $Z$, $\varGamma$ such that $4 \sigma \frac{\optone(Z)}{Z}$ $\ge \varGamma \ge 2 \sigma \frac{\optone(Z)}{Z}$, our algorithm will succeed, and bound the competitive ratio obtained for LP1($Z$).

\begin{lemma}
 If CCFL-LP1($Z$) is feasible and $\frac{\varGamma}{2 \sigma} \ge \frac{\optone(Z)}{Z}$, then \assign does not fail.
 \label{lem:ccflfail1}
\end{lemma}

\begin{proof}
 Let $(\tilde{\bfx}^f, \bty^f, \tl^f)$ be the current values for LP2($Z$, $\varGamma$) and $(\alpha^f, \beta^f, \gamma^f, \delta^f)$ be the current values for D2($Z$, $\varGamma$) when the condition in line~\ref{line:ccflalgofail} is being checked. Let $(\tilde{\bfx}^*, \bty^*, \tl^*)$ and $(\alpha^*, \beta^*, \gamma^*, \delta^*)$ be the optimal primal and dual solutions for LP2($Z$, $\varGamma$) and D2($Z$, $\varGamma$). Then
 
 \begin{equation}
  \opttwo(Z, \varGamma) ~ = ~ Z \tl^* + \sum_i c_i \ty_i^* + \sum_{i,j} a_{ij} x_{ij}^* / \varGamma~ = ~ \sum_j \alpha_j^*
  \label{eqn:ccflfail2}
 \end{equation}
 
\noindent where the second equality is because of LP strong duality. From Lemma~\ref{lem:ccflduals}, $\alpha^f/(\nu \sigma)$ is feasible for the dual. Hence $\sum_j \alpha_j^* \ge \sum_j \alpha_j^f/(\nu \sigma)$. Then from~(\ref{eqn:ccflfail2}), and from Fact~\ref{fact:ccfllp12},

\[
 \optone(Z) ~ = ~ \varGamma \opttwo(Z, \varGamma) ~ \ge ~ \varGamma \sum_j \alpha_j^f/(\nu \sigma) \, .
\]

\noindent Since $\sum_j \alpha_j^f$ is an upper bound on the change in $\cost(\bfx)$ in each phase, $\sum_j \alpha_j^f \ge \cost(\bfx^f) - \cost(\mathbf{0}) - \sum_j \init_j$. Thus

\[
 \nu \sigma \optone(Z)  ~ \ge ~ \varGamma \left( \cost(\bfx^f) - \cost(\mathbf{0}) - \sum_j \init_j \right) \, .
\]

\noindent $\optone(Z)/Z \ge 1$ by Fact~\ref{fact:ccfloptone}. Hence $\varGamma \ge 1$ by the condition in the lemma statement, and thus from Lemma~\ref{lem:ccflinitbound}, and since $\cost(\mathbf{0}) \le 2Z \ln (mn)$,

\[
 \nu \sigma \optone(Z)  ~ \ge ~ \varGamma \left( \cost(\bfx^f) - 2Z - 2Z \ln (mn) \right) \, .
\]

\noindent By definition, $\nu = 1 + \ln (mn) + \tl(\bfx^f)$, and $\tl(\bfx^f) \le \cost(\bfx^f)/Z$. With these substitutions,

\begin{equation}
 \sigma \optone(Z) \left(\ln (emn) + \frac{\cost(\bfx^f)}{Z} \right) ~ \ge ~ \varGamma \left( \cost(\bfx^f) - 2Z - 2Z \ln (mn) \right) \, ,
 \label{eqn:ccflfracoptbound}
\end{equation}

\noindent and from the bound on $\varGamma$ in the lemma statement,

\[
 \frac{1}{2} Z \left(\ln(emn) + \frac{\cost(\bfx^f)}{Z} \right) + 2Z \ln (mn) + 2Z ~ \ge ~ \cost(\bfx^f)\, .
\]

\noindent Simplifying yields $\frac{5}{2} Z \ln (emn) \ge \frac{1}{2} \cost(\bfx^f)$. Hence, $\cost(\bfx^f) \le 5 Z \ln (emn)$ and the algorithm will not fail.
\end{proof}

\begin{lemma}
 If CCFL-LP1($Z$) is feasible and $\frac{\optone(Z)}{Z} \le \frac{\varGamma}{2\sigma} \le 2 \frac{\optone(Z)}{Z}$, \assign  returns a solution to CCFL-LP1($Z$) of cost $O(\ln (mn) \ln (mn \rho)) \optone(Z)$.
 \label{lem:ccflfail2}
\end{lemma}

\begin{proof}
Let $(\bfx^f, \bty^f, \tl^f)$ be the solution for LP2($Z$, $\varGamma$) our algorithm returns. By Lemma~\ref{lem:ccflfail1}, \assign does not fail, and hence $\cost(\bfx^f) \le 5Z \ln (emn)$. Substituting this bound on $\cost(\bfx^f)$ in the expression on the left in~(\ref{eqn:ccflfracoptbound}) yields

\[
 \sigma \optone(Z) \, 6 \ln (emn) ~ \ge ~ \varGamma \left( \cost(\bfx^f) - 2Z - 2Z \ln (mn) \right) \, .
\]

\noindent or $6 \ln (emn) \sigma \optone(Z) + 2 Z\varGamma \ln (emn) \ge \varGamma \cost(\bfx^f)$. Substituting the upper bound on $\varGamma$,

\[
 6  \sigma \ln (emn) \optone(Z) + 8 \sigma \ln(emn)  \optone(Z) ~ \ge ~ \varGamma \cost(\bfx^f) \, .
\]

\noindent Since $\varGamma \cost(\bfx^f)$ is an upper bound on the cost of the solution obtained for CCFL-LP1($Z$), the proof follows.
\end{proof}

The following theorem now follows immediately from Lemmas~\ref{lem:ccflfail1} and~\ref{lem:ccflfail2}.

\begin{theorem}
 If CCFL-LP1($Z$) is feasible and $Z$, $\varGamma$ satisfy $2 \frac{\optone(Z)}{Z} \ge \frac{\varGamma}{2 \sigma} \ge \frac{\optone(Z)}{Z}$, then \assign does not fail and returns a solution to CCFL-LP1($Z$) of cost $O(\ln (mn) \ln (m n \rho)) \optone(Z)$.
 \label{thm:ccfllp1}
\end{theorem}

\subsection{A Doubling Procedure for $\varGamma$.}
\label{sec:facilitydouble}

If we are not given $\varGamma$ that satisfies the conditions of Theorem~\ref{thm:ccfllp1}, we use a doubling procedure similar to that described in Section~\ref{sec:mpcdoubling}. Initially set $\varGamma = 1$ and run \assign. Each execution of \assign is called a \emph{trial}, and each trial $\tau$ has a distinct set of primal and dual variables ($\tilde{\bfx}(\tau)$, $\tilde{\bfv}(\tau)$, $\tilde{\bfw}(\tau)$, $\tl(\tau)$). In each trial, $x_{ij}(\tau)$ is initialized to $x_{ij}^0$ for each client that arrives during that trial, and the other variables are updated accordingly. If a trial fails, we double $\varGamma$ and proceed with a new trial with a new set of primal and dual variables. We continue in this manner, doubling the value of $\varGamma$ after each failure, until all clients are assigned. The cost of the solution we obtain is then at most the sum of the costs obtained in each trial.

Let $\tA(Z, \varGamma)$ be the cost of the solution to LP2($Z$, $\varGamma$) obtained in each trial. Hence, by Fact~\ref{fact:ccfllp12}, $\varGamma \tA(Z, \varGamma)$ is the cost of the solution to CCFL-LP1($Z$) in each trial. Define $\mathcal{A}(Z)$ as the sum of the (partial) solutions to CCFL-LP1($Z$) in each trial. Thus, $\mathcal{A}(Z) := \sum_{\varGamma} \varGamma \tA(Z, \varGamma)$.

\begin{theorem}
 If CCFL-LP1($Z$) is feasible, \assign with the doubling procedure for $\varGamma$ obtains a fractional solution to CCFL-LP1($Z$) of cumulative cost $O(\ln (mn) \ln (m n \rho)) \optone(Z)$ over all trials.
 \label{thm:ccflZ}
\end{theorem}

\begin{proof}
 By assumption, CCFL-LP1($Z$) is feasible. Initially, $\varGamma = 1 \le \optone(Z)/Z$ since $\lambda \ge 1$. Since we double $\varGamma$ each time \assign fails, and by Lemma~\ref{lem:ccflfail1} \assign will not fail for $\varGamma \ge 2 \sigma \optone(Z)/Z$, $\varGamma \le 4 \sigma \optone(Z)/Z$ in any trial. Further, in any successful trial, $\tA(Z, \varGamma) \le 8Z \ln (emn)$. For any failed trial $\tau$, at the beginning of the phase when the trial failed, $\tA(Z, \varGamma) \le \cost(\bfx) \le 6 Z \ln (3mn)$ by Corollary~\ref{cor:ccflphasebegin}. In one phase, each variables $x_{ij}$ is incremented by at most a factor of $\mu = 1 + 1/(6 \ln (emn))$. Thus in any failed trial, $\tA(Z, \varGamma) \le \mu 6Z \ln (emn)$ $\le Z(1 + 6 \ln (emn))$. 
 
 Thus over all trials, the cumulative cost of the solution to CCFL-LP1($Z$) $\mathcal{A}(Z)$ is at most $Z (1 + 6 \ln (emn)) \sum_{\varGamma} \varGamma$. Let $\varGamma^f$ be the value of $\varGamma$ in the final trial. Since $\varGamma$ is doubled after each trial, $\mathcal{A}(Z) \le Z (1 + 6 \ln (emn)) 2 \varGamma^f$. Since $\varGamma^f \le 4 \sigma \optone(Z)/Z$, and $\sigma = 4e^2 \ln (2 \mu mn\rho)$, the theorem follows.
 \end{proof}

\subsection{Obtaining an Integral Solution}
\label{sec:ccflintegral}

We will now build upon the fractional algorithm for CCFL and give a randomized rounding procedure that obtains an integral assignment of clients to facilities. As before, we will assume that the fractional assignment $x_{ij}$ of client $j$ to any facility $i$ not in $F_j(Z)$ is always zero. We first give a rounding procedure that uses Theorem~\ref{thm:ccflZ} and obtains an integral assigment of clients to facilities. We will then show how to use this integral assignment procedure for a fixed parameter $Z$ to obtain an $O(\ln^2(mn) \ln(mn \rho))$-competitive solution to $Z^*$.

We run our rounding procedure whenever a new client $j$ arrives, and given a fractional solution $(\bfx, \bfy, \lambda)$ for CCFL-LP1($Z$) that satisfies $\sum_i x_{ij} \ge 1$. The procedure returns a set of open facilities and an integral assignment of $j$ to an open facility.

We assume $x_{ij} \le 1$ without loss of generality. For each client $j$, let $S(j) := \{i: x_{ij} \ge 1/(2m) \}$. $C_j$ is a set of \emph{candidate} facilities for the integral assignment for client $j$. Initially, $C_j = \emptyset$. $\open$ is the set of facilities opened so far, and $\open = \emptyset$ initially.

Our randomized rounding procedure is as follows. For each facility $i$, select $r = \lceil 4e \ln n\rceil$ random variables uniformly at random between 0 and 1; let $t_{i1}$, $t_{i2}$, $\dots$, $t_{i,r}$ be these random variables for facility $i$, and $\bar{t}_i := \min_k t_{ik}$. When client $j$ arrives,

\noindent \textbf{Step 1:} For each facility $i \not \in \open$, add $i$ to $\open$ if $y_i \ge \bar{t}_i$.

\noindent \textbf{Step 2:} For each facility $i \in S_j$, add $i$ to $C_j$ independently with probability $x_{ij}/y_i$. If $i \in C_j$, then $i$ is a \emph{candidate} for $j$.

\noindent \textbf{Step 3:} For each facility $i \in \open$, assign client $j$ to $i$ if $i$ is a candidate for $j$. Denote an assignment of client $j$ to facility $i$ by $j \rightarrow i$. If $j$ is assigned to multiple facilities, pick one arbitrarily.

\noindent \textbf{Step 4:} If $j$ is not yet assigned to any facility, assign it to facility $i \in S_j$ that minimizes $c_i + a_{ij} + p_{ij}$. 

Steps 3 and 4 thus give an integral assignment of clients to facilities. We show:

\begin{theorem}
 The integral assignment obtained has total cost $O(\ln (mn) \ln (mn \lambda) \ln (mn \rho)) \opt(Z)$. 
 \label{thm:ccflintegral}
\end{theorem}

We first show that with high probability, no client is assigned in Step 4:

\begin{lemma}
For any client, the probability that it is assigned in Step 3 is at least $1 - 1/n^2$.
\label{lem:ccflstep3}
\end{lemma}

\begin{proof}
For a client $j$ that has just arrived, consider a facility $i \in S_j$. In Step 3,

\begin{align*}
\mbox{Pr[$j \rightarrow i$]} & ~ = ~ \mbox{Pr[$i$ is open]} \cdot \mbox{Pr[$i \in C_j$]}  \\
	& ~ = ~ (1-\mbox{Pr[$y_i < \bar{t}_i$]}) ~ \frac{x_{ij}}{y_i} \\	
	& ~ = ~ (1 - (1 - y_i)^r) ~ \frac{x_{ij}}{y_i} \\
	& ~ \ge ~ (1 - e^{-y_i r}) ~ \frac{x_{ij}}{y_i} \\
	& ~ \ge ~ r x_{ij} / e 
\end{align*}

\noindent where the first inequality is because $(1+x) \le e^x$ for all $x \in \mathbb{R}$, and the second inequality is because $e^{-x} \le 1-(x/e)$ for $0 \le x \le 1$. Thus, the probability that $j$ is not assigned to a fixed $i \in S_j$ is at most  $1 - (r x_{ij}/e)$ $\le e^{- r x_{ij}/e}$. Since these probabilities are independent,

\begin{align*}
\mbox{Pr[$j$ is not assigned]} & ~ = ~ \prod_{i \in S_j} \mbox{Pr[$j$ is not assigned to $i$]} \\
	& ~ \le ~ \prod_{i \in S_j} e^{-r x_{ij}/e}  ~ = ~ e^{-r \sum_{i \in S_j} x_{ij}/e} \, .
\end{align*}

For any $i \not \in S_j$, $x_{ij} \le 1/(2m)$, hence $\sum_{i \not \in S_j} x_{ij} \le 1/2$. Thus $\sum_{i \in S_j} x_{ij}\ge 1/2$, and since $r \ge 4 e \ln n$, the probability that client $j$ is not assigned to any facility in Step 3 is at most $1/(n^2)$.
\end{proof}

Since each client is assigned in Step 3 with high probability, the effect of Step 4 on the total cost of the integral assignment is negligible. The following lemma follows immediately from Lemma~\ref{lem:ccflstep3} and since $S_j \subseteq F_j(Z)$:

\begin{lemma}
 The assignments in Step 4 increase the total cost of the integral assignment by at most $Z/n$.
 \label{lem:ccflstep4}
\end{lemma}

We now show bounds on the total cost for assignments in the remaining steps. We first bound the expected fixed-charges and assignment costs.

\begin{lemma}
The expected sum of fixed-charges $\sum_{i \in \open} c_i$ is at most $r \sum_i c_i y_i$.
\label{lem:ccflfixed}
\end{lemma}

\begin{proof}
The probability that facility $i$ is in $\open$ is Pr[$y_i \ge \bar{t}_i$] $\le \sum_{k=1}^r \mbox{Pr [$y_i \ge t_{ik}$]}$ $=r y_i$ by the union bound. The lemma follows.
\end{proof}

\begin{lemma}
 The expected assignment costs for clients assigned in Step 3 is at most $r \sum_{i,j} a_{ij} x_{ij}$.
 \label{lem:ccflassign}
\end{lemma}

\begin{proof}
 For facility $i$, the probability that $i$ is in $\open$ is at most $r y_i$ by the proof of Lemma~\ref{lem:ccflfixed}. For any client $j$ assigned in Step 3, Pr[$j \rightarrow i$] = Pr[$i$ is open] $\cdot$ Pr[$i \in C_j$] $\le r x_{ij}$. The bound on the expected assignment cost follows.
\end{proof}

We now prove the bound on the expected maximum congestion for the integral assignment. Define the \emph{candidate congestion} for a facility $i$ as $L^{(c)}_i := \sum_{j: i \in C_j} p_{ij}$. For any realization of the random bits, the candidate congestion of any facility is an upper bound on the actual congestion for clients assigned to the facility in Step 3. We will prove an upper bound on the expected value of the maximum candidate congestion over all facilities, which will give us a bound on the expected value of the maximum congestion.

We consider the maximum candidate congestion instead of the actual maximum congestion because for a fixed facility $i$, the actual assignments of the clients are not independent of each other. If client $j-1$ is assigned to facility $i$, then the facility must be open, and hence the probability that client $j$ is assigned to $i$ increases. However, for any facility $i$ and clients $j \neq j'$, Pr[$i \in C_j$] and Pr[$i \in C_{j'}$] are independent.

For the next lemma, for any client $j$ that arrived in the current trial, define $y_i(j)$ as the value of $y_i$ when the randomized rounding procedure was executed for client $j$.

\begin{lemma}
The candidate congestion $L^{(c)}_i$ on any facility $i$ at most $Z \ln (2 e m \lambda)$ in expectation.
\label{lem:facilityimakespan}
\end{lemma}

\begin{proof}
Consider a fixed open facility $i$.
\[
E[L^{(c)}_i] ~ = ~ \sum_{j: i \in S_j} \frac{p_{ij} x_{ij}}{y_i(j)} \, .
\]

For any client $j$, 

\begin{equation}
 Z ~ \ge \frac{1}{y_i(j)} \sum_{j' \le j} p_{ij'} x_{ij'} \, .
 \label{eqn:facilityZprime}
\end{equation}

By Lemma~\ref{lem:techPT} with $P = \sum_{j:i \in S_j} \frac{p_{ij} x_{ij}}{y_i(j)}$ and $Z$ as defined here,

\[
E[L^{(c)}_i] ~ = ~ P ~ \le ~ Z \left(1 + \ln \frac{y_i(n)}{y_i(k)} \right) \, .
\]

\noindent where $k$ is the first client $j$ such that $i \in S_j$. Then $y_i(k) \ge x_{ik} \ge 1/(2m)$, and $y_i(n) \le \lambda$. The lemma follows.
\end{proof}

To bound the maximum candidate congestion, we use the following inequality:

\begin{lemma}[\cite{habib1998}]
 Let $X_1, \dots, X_n$ be independent random variables with Pr($X_j = 1$) $=q_j$, Pr($X_j=0$) $=1-q_j$.  For $X = \sum_{j=1}^n a_j X_j$, define $\nu = \sum_{j=1}^n a_j^2 q_j$ and $a = \max_j a_j$. Then
 \[
  \mbox{Pr($X > E(X) + \mu$)} \le e^{- \frac{\mu^2}{2 \nu + \frac{2a\mu}{3}}}
 \]
 \label{lem:probineq}
\end{lemma}

\begin{lemma}
 The maximum candidate congestion is at most $4Z \ln (2em \lambda)$ in expectation.
 \label{lem:ccflcongestion}
\end{lemma}

\begin{proof}
 Fix facility $i$. For each client $j$, let $a_j = p_{ij}/Z$ if $i \in S_j$, and $a_j = 0$ otherwise. Hence $a_j \le 1$. Define random variable $X_j = 1$ if $i \in C_j$, and $X := \sum_j a_j X_j$ $= \frac{L^{(c)}_i}{Z}$, since $C_j \subseteq S_j$. Let $q_j := \mbox{Pr[$i \in C_j$]}$. From Lemma~\ref{lem:facilityimakespan}, $E(X) \le \ln (2em\lambda)$.
 
 We will use Lemma~\ref{lem:probineq} to show that with high probability, the candidate congestion $L^{(c)}_i$ on facility $i$ does not exceed thrice the expected value. Let $\mu := 3\ln (2em\lambda)$. Then from Lemma~\ref{lem:probineq},
 
 \[
  \mbox{Pr $\left(X > 3 \ln (2em\lambda)\right)$} ~ \le ~ e^{- \frac{\mu^2}{2 \nu + \frac{2a\mu}{3}}}
 \]

\noindent Since $a_j \le 1$ for all $j$, $a \le 1$. Also, $\nu = \sum_j a_j^2 q_j$ $\le \sum_j a_j q_j$ $= E(X) = \mu/3$. Then

 \begin{align}
  \mbox{Pr $\left(X > 3 \ln (2em\lambda)\right)$} & ~ \le ~ e^{- \frac{\mu^2}{\frac{2\mu}{3} + \frac{2\mu}{3}}} \nonumber \\
    & ~ \le ~ e^{- \frac{3\mu}{4}} \nonumber \\
    & ~ \le ~ (2em \lambda)^{-2} ~ \le ~ \frac{1}{4m^2} \label{eqn:fccintcong}
 \end{align}

\noindent where the last inequality is because $\lambda \ge 1$ by CCFL-LP1($Z$). Thus, the probability that the candidate congestion of a fixed facility exceeds $3 Z \ln (2em\lambda)$ is at most $1/(4m^2)$, and by the union bound, the probability that the candidate congestion of any facility exceeds $3 Z \ln (2em\lambda)$ is at most $1/(4m)$. To get the bound on the expectation of the maximum candidate congestion, we observe that the candidate congestion of any facility can at most be $\sum_{j:i \in S_j} p_{ij}$ $\le 2m \lambda \sum_j p_{ij} x_{ij} / y_i(j)$, since $x_{ij} \ge 1/(2m)$ for any client $j$ with $i \in S_j$, and $y_i(j) \le \lambda$. From the constraints in CCFL-LP1($Z$), $\sum_j p_{ij} x_{ij} / y_i(j) \le Z$, and hence the candidate congestion is bounded by $2m\lambda Z$. Thus the expected value of the maximum candidate congestion is at most $3 Z \ln (2em\lambda) (1-1/(4e^2m^2\lambda^2)) + 2m\lambda Z/(4e^2m^2\lambda^2) \le 4Z \ln (2em\lambda)$ using $\lambda \ge 1$.
\end{proof}

We now use the bounds on the congestion, fixed-charges and assignment costs to prove the bound on the expected total cost from Theorem~\ref{thm:ccflintegral}.

\vspace{0.2in}
\noindent \emph{Proof of Theorem~\ref{thm:ccflintegral}.}
By Theorem~\ref{thm:ccfllp1} and Lemmas~\ref{lem:ccflfixed} and~\ref{lem:ccflassign}, the sum of the fixed-charges and assignment costs for the integral assignments is $O(\ln n \ln (mn) \ln (mn\rho))$ $\opt(Z)$ in expectation. By Lemma~\ref{lem:ccflcongestion}, the maximum congestion is $O(\ln (m \lambda))Z$ $\le O(\ln m) \optone(Z)$ in expectation, since $\optone(Z) \ge Z \lambda$. Assignments in Step 4 add at most $Z/n$ to the total cost by Lemma~\ref{lem:ccflstep4}. Summing up, the total cost of the integral assignment is $O(\ln (mn \lambda) \ln (mn) \ln (mn\rho))$ $\optone(Z)$ in expectation.
\qed

\subsection{A Doubling Procedure for $Z$}
\label{sec:ccflZdouble}

We now use the rounding procedure with a doubling argument and describe an algorithm for the CCFL problem. We assume we are given a procedure that, given $Z$, maintains an integral assigment of clients to facilities of total cost $O(\ln^2 (mn) \ln (mn\rho))$ $\optone(Z)$. Let $\mathbb{Q}(Z)$ denote this procedure, and let $c_\mathbb{Q}(Z)$ be the expected total cost obtained by this procedure. We start with $Z = \min_i \{c_i + p_{ij} + a_{ij}\}$, and run $\mathbb{Q}(Z)$. Call each execution of $\mathbb{Q}(Z)$ an \emph{epoch}. In each epoch, $\mathbb{Q}(Z)$ uses a distinct set of variables. If $c_\mathbb{Q}(Z)$ $\ge O(\ln^2(mn) \ln (mn\rho)) Z$ in an epoch, or if for some client $j$ $F_j(Z) = \emptyset$, we fail $\mathbb{Q}(Z)$, double the value of $Z$, and run $\mathbb{Q}(Z)$ for the clients not yet assigned with the new value of $Z$ and a new set of variables.

We start by showing that if $Z$ is at least $Z^*$, then $Z$ and $\opt(Z)$ are close, and bounding the total cost returned by the procedure $\mathbb{Q}(Z)$.

\begin{lemma}
 If $Z \ge Z^*$, then $\opt(Z) \le 2Z$ and $c_{\mathbb{Q}}(Z)$ $\le O(\ln^2(mn) \ln (mn \rho)) Z$.
 \label{lem:ccflboundZ}
\end{lemma}

\begin{proof}
 If $Z \ge Z^*$, then CCFL-LP1($Z$) is feasible, since for every client $j$, $F_j(Z) \neq \emptyset$. Consider the solution to CCFL-LP1($Z$) that sets $y_i = 1$ for every facility that is open in the optimal solution, and $x_{ij}=1$ if client $j$ is assigned to facility $i$ in the optimal solution. Set $\lambda = 1$. Since $Z^*$ is an upper bound on the congestion of any facility in the optimal solution and $Z \ge Z^*$, this gives a feasible solution to CCFL-LP1($Z$). Then the sum of the assignment costs and fixed charges for this solution are at most $Z^*$. Also, $Z \lambda = Z$. Hence, the total cost of this solution is at most $Z + Z^*$ $\le 2Z$, and hence $\optone(Z) \le 2Z$. The bound on $c_\mathbb{Q}(Z)$ follows from Theorem~\ref{thm:ccflintegral}.
\end{proof}

Finally, using $\mathbb{Q}(Z)$ with the doubling argument described,

\begin{theorem}
 The sum over all epochs of the expected total cost in each epoch is $O(\ln^2(mn) \ln (mn \rho)) Z^*$.
 \label{thm:ccflend}
\end{theorem}

\begin{proof}
 Let $Z^f$ be the value of $Z$ in the final epoch. By Lemma~\ref{lem:ccflboundZ} and our failure condition for an epoch, if $Z^f \ge Z^*$ then the epoch must succeed. Since $Z$ is doubled after every failed epoch, $Z^f \le 2 Z^*$. For every epoch, $c_{\mathbb{Q}}(Z) \le O(\ln^2(mn) \ln (m n \rho)) Z$,  since otherwise we would have failed the current epoch. Since we double $Z$ on every failed epoch, the sum over all epochs of the total cost in each epoch is $\sum_Z c_{\mathbb{Q}}(Z)$ $\le O(\ln^2(mn) \ln (m n \rho)) \sum_Z Z$, where the sum is over the values of $Z$ in each epoch. Since $Z$ is doubled after each epoch, $\sum_Z Z \le 2 Z^f \le 4 Z^*$. Hence, the total cost over all epochs of the integral assignment obtained by our algorithm is $O(\ln^2(mn) \ln (m n \rho)) Z^*$.
 \end{proof}

\subsection{A Lower Bound for UMSC and CCFL}
\label{sec:machinebad}

In this section we give a lower bound on the competitive ratio for bicriteria results for UMSC. These lower bounds on bicriteria results motivate the problem of minimizing the sum of makespan and startup costs for machine scheduling that we study in Section~\ref{sec:ccfl}.

CCFL generalizes UMSC, and our lower bound extends to bicriteria results for CCFL as well. Let $T^*$ be the makespan of an assignment of jobs to machines, $m$ and $n$ be the number of machines and jobs respectively, and $\rho$ as defined in Section~\ref{sec:ccfl}. Let $C^*$ be the optimal startup cost of an assignment with makespan $T^*$. We show

\begin{theorem}
 No deterministic online algorithm can obtain a solution with makespan $o(m)T^*$ and startup cost within a factor polylogarithmic in $m$, $n$, and $\rho$ of $C^*$, even if $T^*$ is available offline.
\end{theorem}

Our lower bound is in fact for any fractional solution that allows jobs to be assigned fractionally to machines, and machines to be fractionally opened. For each job, the sum of the fractional assigments to machines must be at least 1, and the fraction by which any machine is opened must be at least the fractional assignment of any job to the machine. 

Our example has $m$ machines, and $2(m-1)$ jobs. We choose $T^* = m$. The sequence of job arrivals is fixed, and known to the online algorithm; the only freedom we allow is that we could stop sending jobs from the sequence at any time, and send trivial jobs with $p_{ij} = 0$ for all $i$ instead. We will show that no deterministic online algorithm can obtain a startup cost that exceeds the optimal startup cost by factor at most polynomial in $m$, and a makespan at most $T^* m/2$, even in a fractional solution.

The cost of machine $i$ in our example is $e^{m(i-1)}$. Thus machine 1 has cost 1, and machine $m$ has cost $e^{m^2-m}$. The index of each job corresponds to its arrival in the sequence. Thus job $j$, if it arrives, is the $j$th job to arrive. The jobs are of two types, \emph{even} jobs and \emph{odd} jobs, corresponding to their index. An odd job $j$ can be assigned to either of two machines: machine 1 with processing time $T^*$, or machine $(j+3)/2$ with processing time $\epsilon$. 
An even job $j$ can only be assigned to machine $(j+2)/2$. 

\begin{table}[h]
\centering
 \begin{tabular}{|c|c|c|c|c|c|c|}
 \hline
  Machines & 1 & 2 & 3 & 4 & 5 \\ \hline
  Job 1 & $T^*$ & $\epsilon$ & & & \\ \hline
  Job 2 &  & $T^*-\epsilon$ & & & \\ \hline
  Job 3 & $T^*$ & & $\epsilon$ & & \\ \hline
  Job 4 &  & & $T^*-\epsilon$ & & \\ \hline
  Job 5 & $T^*$ & & & $\epsilon$ & \\ \hline
  Job 6 &  & & & $T^*-\epsilon$ & \\ \hline
  Job 7 & $T^*$ & & & & $\epsilon$ \\ \hline
  Job 8 &  & & & & $T^*-\epsilon$ \\ \hline
 \end{tabular}
 \caption{The example depicted for 5 machines. All missing entries are $\infty$.}
\end{table}

We start with an observation about the optimal assignment.

\begin{lemma}
 For jobs $1, \dots, k$ with $k$ even, there is an assignment of these jobs on machines $2, \dots, (k+2)/2$ with makespan $T^*$.
 \label{lem:even}
\end{lemma}

\begin{proof}
 Assign each odd job $j$ to machine $(j+3)/2$, and each even job $j$ to machine $(j+2)/2$. Then no job is assigned to machine 1, and no job gets assigned to machine $i$ for $i > (k+2)/2$. Also, each machine $2 \le i \le (k+2)/2$ gets assigned at most one odd job with processing time $\epsilon$, and one even job with processing time $T^* - \epsilon$. The lemma follows.
\end{proof}

Suppose now that an odd job $j$ arrives. Then

\begin{lemma}
 Any fractional algorithm that obtains startup cost $o(e^m/m)$ of the optimal startup cost must assign job $j$ to machine 1 with fractional value $\ge 1/2$.
 \label{lem:odd}
\end{lemma}

\begin{proof}
 Let $k:= (j+3)/2$. By Lemma~\ref{lem:even}, jobs $1, \dots, j-1$ can be assigned to machines $2$, $\dots$, $k-1$ with makespan $T^*$. Job $j$ can be assigned to either machine 1 with processing time $T^*$, or machine $k$ with processing time $\epsilon$. Thus in an optimal assignment of jobs $1, \dots, j$, these jobs can be assigned to machines $1, \dots, k-1$ with makespan $T^*$. The online algorithm cannot assign $j$ to machine $k$ with fractional value $\ge 1/2$, since the startup cost would exceed $e^{km}/2$, while the startup cost for the optimal assignment is at most $m e^{(k-1)m}$. Thus, the algorithm must assign $j$ to machine 1 with fractional value $\ge 1/2$.
\end{proof}

From Lemma~\ref{lem:odd}, every odd job must be assigned to machine 1 with fractional value $\ge 1/2$, and hence the makespan obtained must be at least $m/2$, since in our example we send $2(m-1)$ jobs. Thus, any deterministic algorithm must have makespan at least $m/2$, or startup cost $\Omega(e^m/m)$.

%% file: of-appendix.tex
\section{Appendix}
\label{sec:appendix}

We present technical lemmas which are used in various proofs.

\begin{lemma}
 Let $y := \sum_{i=1}^n r_i$, where $0 < r_i \le 1$ for $i \in [n]$, and $\prod_{i=1}^n r_i = P$. Then $y$ is minimized when $r_i = P^{1/n}$ $\forall i$, and the minimum value is $nP^{1/n}$. 
 \label{lem:hessian}
\end{lemma}

\begin{proof}
 The proof is by induction on $n$. For $n=1$, the lemma is obviously true. Let $\gamma_k(P)$ be the minimum value of the sum of $k$ variables, when the product of the variables is $P$. Then  $\gamma_n(P)$ $= \min_{0< r_1 \le 1} \{r_1 + \gamma_{n-1}(P/r_1)\}$. By the inductive hypothesis, $\gamma_{n-1}(P/r_1) = (n-1)(P/r_1)^{1/(n-1)}$. Hence
 
 \[
  \gamma_n(P) ~ = ~ \min_{0< r_1 \le 1} \left\{r_1 + (n-1)\left( \frac{P}{r_1}\right)^{1/(n-1)} \right\} \, .
 \]

 We will show that the expression on the right is minimized when $r_1 = P^{1/n}$. Then by the inductive hypothesis, each of the other variables is $P^{1/n}$ as well, completing the proof.
 
 Let $z = r_1 + (n-1)\left( \frac{P}{r_1}\right)^{1/(n-1)}$. Then
 
 \[
  \frac{dz}{dr_1} ~ = ~ 1 - \frac{1}{n-1} r_1^{-n/(n-1)} (n-1) P^{1/(n-1)} \, ,
 \]

\noindent and setting $dz/dr_1 = 0$, we obtain $r_1 = P^{1/n}$. Further, $d^2z/dr_1^2 \ge 0$ $\forall r_1 \ge 0$. Hence, the point $r_1 = P^{1/n}$ is a minimum. This completes the proof.
\end{proof}

\begin{lemma}
  For any $n \in \mathbb{Z}_+$ and $a_1$, $a_2$, $\dots$, $a_n \in \mathbb{R}_{\ge 0}$ with $a_1 > 0$,
 
\[
 \sum_{i \in [n]} \frac{a_i}{\sum_{j \le i} a_j} ~ \le ~ 1 + \ln \frac{\sum_{i=1}^n a_i}{a_1}
\]
\label{lem:genineq2}
\end{lemma}

\begin{proof}
  For $n=1$, the statement is trivially true. For $n \ge 2$, define $b_i = \sum_{j \le i} a_j$. Then $a_i = b_i - b_{i-1}$ for $i \ge 2$, and hence

\begin{eqnarray}
\sum_{i \in [n]} \frac{a_i}{\sum_{j \le i} a_j} ~ = ~ 1 + \sum_{i=2}^n \frac{b_i - b_{i-1}}{b_i} ~ = ~ 1 + \sum_{i=2}^n \left( 1 - \frac{b_{i-1}}{b_i} \right) \label{eqn:genineq2:withb}
\end{eqnarray}

\noindent Let $r_i = \frac{b_i}{b_{i+1}}$, and let $y = \sum_{i=2}^n (1 - \frac{b_{i-1}}{b_i})$, then $y = \sum_{i=1}^{n-1} (1-r_i)$ $= (n-1) - \sum_{i=1}^{n-1} r_i$. Since each $r_i \le 1$ and $\prod_{i=1}^{n-1} r_i = b_1/b_n$, by Lemma~\ref{lem:hessian}, 

\begin{eqnarray}
 y & \le & (n-1) - (n-1) \left(\frac{b_1}{b_n} \right)^{1/(n-1)} \label{eqn:genineq2:y} \, .
\end{eqnarray}

\noindent Let $c = \frac{b_1}{b_n}$ and $z = (n-1) - (n-1)c^{1/(n-1)}$. Differentiating $z$ w.r.t. $n$,

\begin{eqnarray}
 \frac{\partial z}{\partial n} & = & 1 - c^{1/(n-1)} +  \frac{c^{1/(n-1)}}{n-1} \ln c \nonumber
\end{eqnarray}

\noindent and again $\partial^2 z/\partial n^2 = -c^{1/(n-1)} \ln^2c/(n-1)^3 < 0$. Hence, $z$ is maximized when $(n-1) - (n-1) c^{1/(n-1)} = c^{1/(n-1)} \ln \frac{1}{c}$. Substituting the expression on the left in this equality in~(\ref{eqn:genineq2:y}) gives us $y \le  c^{1/(n-1)} \ln \frac{1}{c}$, and since $c = \frac{b_1}{b_n} \le 1$, 

\begin{eqnarray}
y \le \ln \frac{1}{c} = \ln \frac{b_n}{b_1} = \ln \frac{\sum_{i=1}^n a_i}{a_1} \, . \label{eqn:genineq2:ylast}
\end{eqnarray}

\noindent Then from~(\ref{eqn:genineq2:withb}) and~(\ref{eqn:genineq2:ylast}), and by definition of $y$,

 \begin{eqnarray}
 \sum_{i \in [n]} \frac{a_i}{\sum_{j \le i} a_j} ~ \le ~ 1 + \ln \frac{\sum_{i=1}^n a_i}{a_1} \, . \nonumber
\end{eqnarray}
\end{proof}

Lemma~\ref{lem:rjbound} is used in the proof of Lemma~\ref{lem:primaldual} in Section~\ref{sec:algo}. As in Section~\ref{sec:algo}, $\mu := 1+ \frac{1}{3 \ln (em)}$. 

\begin{lemma}
 Given $\bfx'$ and $\bfx''$ with $\tl(\bfx') \le 3\ln (e m)$ and $x_j' \le x_j'' \le \mu x_j'$ where $\mu = 1 + \frac{1}{3\ln (e m)}$, $r_j(\bfx'') \le e r_j(\bfx')$.
 \label{lem:rjbound}
\end{lemma}

\begin{proof}
 By definition of $r(\bfx)$ in~(\ref{eqn:rj}), and since $x_j' \le x_j'' \le \mu x_j'$,
 
\begin{equation}
 r_j(\bfx'') ~ = ~ \displaystyle \frac{\sum_{k \in [m]} \tilde{p}_{kj} \exp(\tpx'')_k}{\sum_{k \in [m]} \exp(\tpx'')_k} ~ \le ~ \frac{\sum_{k \in [m]} \tilde{p}_{kj} \exp(\mu (\tpx')_k)}{\sum_{k \in [m]} \exp(\tpx')_k} \, .
 \label{eqn:rjbound1}
\end{equation}

\noindent Since $\tl(\bfx') \le 3\ln (e m)$, $\forall k$, $(\tpx')_k \le 3\ln (e m)$, and hence $\mu (\tpx')_k$  $= (\tpx')_k + (\tpx')_k/ (3 \ln (e m))$ $\le (\tpx')_k + 1$. Substituting $ (\tpx')_k + 1$ for $\mu(\tpx')_k$ in~(\ref{eqn:rjbound1}) yields

\[
 r_j(\bfx'') ~ \le ~ \frac{\sum_{k \in [m]} \tilde{p}_{kj} \exp((\tpx')_k + 1)}{\sum_{k \in [m]} \exp(\tpx')_k} ~ = ~ e \, \frac{\sum_{k \in [m]} \tilde{p}_{kj} \exp(\tpx')_k}{\sum_{k \in [m]} \exp(\tpx')_k} ~ = ~ e \tr_j(\bfx') \, ,
\]

\noindent proving the lemma.
\end{proof}

Lemma~\ref{lem:ccfltechrate} is used in the proof of Lemma~\ref{lem:ccflincr} in Section~\ref{sec:ccfl}. Define $\mu := 1 + \frac{1}{6 \ln (emn)}$. We assume that $\bfx'$, $\bfx''$ satisfy $x_{ij}' \le x_{ij}'' \le \mu x_{ij}'$ for all $i$, $j$; that $\cost(\bfx') \le 6 Z \ln (emn)$,  and that $\varGamma \tw_i(\bfx') = x_{ij}'$ iff $\varGamma \tw_i(\bfx'') = x_{ij}''$ for all $i$, $j$.

\begin{lemma}
 For all $i,j$, $\rate_{ij}(\bfx'') \le e \, \rate_{ij}(\bfx')$.
 \label{lem:ccfltechrate}
\end{lemma}

\begin{proof}
Given $\bfx'$, $\bfx''$ as defined in the lemma, and a fixed facility $i$, we will show that 

\begin{equation}
 \frac{\exp(\sum_{j'} l_{ij'} x_{ij'}''/(Z \varGamma))}{\sum_{i'} \exp(\sum_{j'} l_{i'j'} x_{i'j'}''/(Z \varGamma))} ~ \le ~ e \, \frac{\exp(\sum_{j'} l_{ij'} x_{ij'}'/(Z \varGamma))}{\sum_{i'} \exp(\sum_{j'} l_{i'j'} x_{i'j'}' /(Z \varGamma))} \, , \label{eqn:techs1}
\end{equation}

\noindent and 

\begin{equation}
 \frac{e^{x_{ij}''/\varGamma}}{\sum_{i', j'} e^{x_{i'j'}''/\varGamma}} ~ \le ~ e \,  \frac{e^{x_{ij}'/\varGamma}}{\sum_{i', j'} e^{x_{i'j'}'/\varGamma}} \label{eqn:techs2} \, .
\end{equation}

\noindent Since the other terms in $\rate_{ij}(\bfx)$ are constants, this will prove the lemma.

Since $Z \tl(\bfx') \le \cost(\bfx')$, and $\cost(\bfx') \le 6 Z \ln (emn)$ by the condition in the lemma statement, $\tl(\bfx') \le 6 \ln (emn)$. Hence $\sum_{j'} l_{ij'} x_{ij'}' / (Z \varGamma) \le \tl(\bfx') \le 6 \ln(emn)$, and $x_{ij}'/\varGamma \le \tl(\bfx') \le 6 \ln (emn)$. Thus

\begin{equation}
 \frac{\sum_{j'} l_{ij'} x_{ij'}''}{Z \varGamma} ~ \le ~ \mu \frac{\sum_{j'} l_{ij'} x_{ij'}'}{Z \varGamma} ~ = ~ \left(1 + \frac{1}{6 \ln (emn)}\right) \frac{\sum_{j'} l_{ij'} x_{ij'}'}{Z \varGamma} ~ \le ~ \frac{\sum_{j'} l_{ij'} x_{ij'}'}{Z \varGamma} + 1
\label{eqn:technum1}
\end{equation}

\noindent and

\begin{equation}
\frac{x_{ij}''}{\varGamma} ~ \le ~ \mu \frac{x_{ij}'}{\varGamma} ~ = ~ \left( 1 + \frac{1}{6 \ln (emn)} \right) \frac{x_{ij}'}{\varGamma} ~ \le ~ \frac{x_{ij}'}{\varGamma} + 1 \, .
\label{eqn:technum2}
\end{equation}

Then~(\ref{eqn:techs1}) and~(\ref{eqn:techs2}) follow from~(\ref{eqn:technum1}) and~(\ref{eqn:technum2}) respectively.
\end{proof}

For the next lemma, we are given $\mathbf{p}, \mathbf{x}$ and $\mathbf{u} \in \mathbb{R}_+^n$, with the elements of $\mathbf{u}$ non-decreasing. For $1 \le k \le n$, define

\[
T_k  ~ := ~ \displaystyle \frac{1}{u_k}  \sum_{j=1}^k p_j x_j
\]

\noindent and $T := \max_k T_k$. We also define

\[
P ~ := ~ \displaystyle \sum_{j=1}^n \frac{p_j x_j}{u_j} \, .
\]

\begin{lemma}
With $P$, $T$ and $\bfu$ defined as above, $P ~ \le ~ T \left(1 + \ln \frac{u_n}{u_1} \right)$.
\label{lem:techPT}
\end{lemma}

\begin{proof}
We first obtain a different expression for $P$, and then relate the expression we obtain to $T$.

\begin{align}
P & ~ = ~ \sum_{j=1}^n \frac{p_j x_j}{u_j}  ~ = ~ \sum_{j=1}^n \frac{p_j x_j}{u_n} + \sum_{j=1}^n p_j x_j \left( \frac{1}{u_j} - \frac{1}{u_n} \right) \nonumber \\
	& ~ = ~ T_n +  \sum_{j=1}^n p_j x_j \sum_{k=j}^{n-1} \left( \frac{1}{u_k} - \frac{1}{u_{k+1}} \right) ~ = ~ T_n + \sum_{k=1}^{n-1} \sum_{j=1}^k p_j x_j  \left( \frac{1}{u_k} - \frac{1}{u_{k+1}} \right) \nonumber \\
	& ~ = ~ T_n + \sum_{k=1}^{n-1} \left( 1 - \frac{u_k}{u_{k+1}} \right) \frac{1}{u_k} \sum_{j=1}^k p_j x_j   ~ = ~ T_n + \sum_{k=1}^{n-1} \left( 1 - \frac{u_k}{u_{k+1}} \right) T_k \nonumber \\
	& ~ \le ~ T \left( 1 + \sum_{k=1}^{n-1} \left( 1 - \frac{u_k}{u_{k+1}} \right) \right) \label{eqn:techpxu1}
\end{align}

The expression on the right in~(\ref{eqn:techpxu1}) is exactly the same as the expression on the right in~(\ref{eqn:genineq2:withb}), with $b_{i-1} = u_k$. Then from~(\ref{eqn:genineq2:ylast}), since $c = u_1/u_n$ and $y = \sum_{k=1}^{n-1} \left(1 - \frac{u_k}{u_{k+1}}\right)$, 

\[
P  ~ \le ~ T \left( 1 + \ln \frac{u_n}{u_1} \right) \, .
\]
\end{proof}